\documentclass[a4paper,onecolumn,final,
fleqn]{quantumarticle}
\pdfoutput=1
\usepackage[utf8]{inputenc}
\usepackage[backend=biber, style=alphabetic, maxbibnames=10]{biblatex}
\bibliography{refs.bib}
\RequirePackage{amsmath}
\usepackage{tikz,stmaryrd,amssymb,url,color,relsize,mathtools,enumitem,amsthm}
\usepackage[ruled,vlined,linesnumbered, procnumbered]{algorithm2e}
\SetAlgorithmName{Experiment}{experiment}{List of Experiments}
\usepackage{xspace}
\usepackage{xcolor}
\usepackage{dsfont}
\usepackage{hyperref} 
\hypersetup{
	pdftex=true,
	breaklinks=true,
	colorlinks=true,
	linkcolor=dblue,
	urlcolor=dblue,
	citecolor=dgreen
}
\definecolor{dblue}{rgb}{0,0,.8}

\usepackage{xifthen}

\newtheorem{theorem}{Theorem}[section]
\newtheorem{lemma}[theorem]{Lemma}
\newtheorem{corollary}[theorem]{Corollary}
\newtheorem{definition}[theorem]{Definition}

\newtheorem{construction}[theorem]{Construction}
\newtheorem{example}[theorem]{Example}

\DeclareMathOperator{\Tr}{Tr}
\DeclareMathOperator{\id}{id}
\DeclareMathOperator{\negl}{negl}
\DeclareMathOperator*{\E}{\mathbb{E}}

\newcommand{\from}{\leftarrow}
\newcommand{\eps}{\varepsilon}
\newcommand{\ket}[1]{|#1\rangle}

\newcommand{\kb}[2]{\ket{#1}\bra{#2}}
\newcommand{\bra}[1]{\langle#1|}
\newcommand{\dn}[1]{\nm{#1}_\diamond}
\newcommand{\nm}[1]{\left\|#1\right\|}
\newcommand{\lr}[1]{\langle #1 \rangle}
\newcommand{\Trr}[2][]{\Tr\ifthenelse{\isempty{#1}}{}{_{#1}}\left[#2\right]}
\newcommand{\Prr}[1]{\Pr\left[#1\right]}

\newcommand{\QPT}{\axname{QPT}\xspace}
\newcommand{\NM}{\axname{NM}\xspace}
\newcommand{\CiNM}{\axname{CiNM}\xspace}
\newcommand{\DNS}{\axname{DNS}\xspace}
\newcommand{\PNM}{\axname{PNM}\xspace}

\newcommand{\CNM}{\axname{CNM}\xspace}

\newcommand{\QCNM}{\axname{QCNM}\xspace}

\newcommand{\SKQES}{\axname{SKQES}\xspace}
\newcommand{\PKQES}{\axname{PKQES}\xspace}

\newcommand{\PKES}{\axname{PKES}\xspace}

\newcommand{\CPTNI}{\axname{CPTNI}\xspace}
\newcommand{\CPTP}{\axname{CPTP}\xspace}

\newcommand{\IND}{\ensuremath{\mathsf{IND}}\xspace}

\newcommand{\INDparCCAA}{\ensuremath{\mathsf{IND\mbox{-}parCCA2}}\xspace}
\newcommand{\parCCAA}{\ensuremath{\mathsf{parCCA2}}\xspace}

\newcommand{\QCNMreal}{\axname{QCNM\mbox{-}Real}\xspace}
\newcommand{\QCNMideal}{\axname{QCNM\mbox{-}Ideal}\xspace}
\newcommand{\CNMreal}{\axname{CNM\mbox{-}Real}\xspace}
\newcommand{\CNMideal}{\axname{CNM\mbox{-}Ideal}\xspace}
\newcommand{\sckQCNMreal}{\axname{sckQCNM\textbf{-}Real}\xspace}
\newcommand{\sckQCNMideal}{\axname{sckQCNM\textbf{-}Ideal}\xspace}

\newcommand{\Qu}{\ensuremath{\mathsf{Qu}}\xspace}
\newcommand{\Cl}{\ensuremath{\mathsf{Cl}}\xspace}
\newcommand{\Hyb}{\ensuremath{\mathsf{Hyb}}\xspace}

\newcommand{\pk}{\ensuremath{\mathsf{pk}}\xspace}

\newcommand{\KeyGen}{\axname{KeyGen}\xspace}
\newcommand{\Enc}{\axname{Enc}\xspace}
\newcommand{\Dec}{\axname{Dec}\xspace}

\newcommand{\HH}{\mathcal{H}}
\newcommand{\MA}{\mathcal{A}}
\newcommand{\MB}{\mathcal{B}}

\newcommand{\MD}{\mathcal{D}}

\newcommand{\adver}{\ensuremath{\MA}\xspace}

\newcommand{\ch}[2][]{{\color{gray}\ifthenelse{\isempty{#1}}{#2}{#1\rightarrow #2}}}

\newcommand{\I}{\mathds{1}}
\newcommand{\N}{\mathbb{N}}

\newcommand{\C}{\mathbb{C}}
\newcommand{\Z}{\boldsymbol{0}}

\newcommand{\axname}[1]{\ensuremath{\mathsf{#1}}}

\definecolor{dred}{rgb}{.7,.05,.05}
\definecolor{dgreen}{rgb}{.1,.5,.1}

\begin{document}
\title{Non-malleability for quantum public-key encryption}
\author{Christian Majenz}
\affiliation{Institute for Logic, Language and Computation, University of Amsterdam \& QuSoft, Amsterdam, Netherlands}
\email{c.majenz@uva.nl}
\author{Christian Schaffner}
\affiliation{Institute for Logic, Language and Computation, University of Amsterdam \& QuSoft, Amsterdam, Netherlands}
\email{c.schaffner@uva.nl}
\author{Jeroen van Wier}
\affiliation{Interdisciplinary Centre for Security, Reliability and Trust, University of Luxembourg, Esch-sur-Alzette, Luxembourg}
\email{jeroen.vanwier@uni.lu}
\maketitle

\begin{abstract}
	\hspace{-1.6pt}Non-malleability is an important security property for public-key encryption (PKE). Its significance is due to the fundamental unachievability of integrity and authenticity guarantees in this setting, rendering it the strongest integrity-like property achievable using only PKE, without digital signatures. In this work, we generalize this notion to the setting of \emph{quantum} public-key encryption. Overcoming the notorious ``recording barrier'' known from generalizing other integrity-like security notions to quantum encryption, we generalize one of the equivalent classical definitions, comparison-based non-malleability, and show how it can be fulfilled. In addition, we explore one-time non-malleability notions for symmetric-key encryption from the literature by defining plaintext and ciphertext variants and by characterizing their relation.
\end{abstract}
\section{Introduction}
\label{sec:intro}

The development of quantum information processing technology has accelerated recently, with many large public and private players investing heavily \cite{EU-is-slow}. A future where communication networks include at least some high-capacity quantum channels and fault-tolerant quantum computers seems therefore more and more likely. How will we secure communication over the resulting ``quantum internet'' \cite{Wehner2018}? One approach is to rely on features inherent to quantum theory to get unconditional security, e.g. by using teleportation. Such methods are, however, a far cry from the classical standard internet cryptography in terms of efficiency, as they require interaction. A different and more efficient approach is to generalize modern private- and public-key cryptography to the quantum realm.

In this paper, we focus on the notion of non-malleability, which captures the idea that an encrypted message cannot be altered by a third party in a structured manner. 
This notion, first introduced by Dolev, Dwork and Naor \cite{Dolev2003}, derives its importance from the fact that it is the strongest integrity-like notion that is achievable using public-key encryption only. The aim of this work is to generalize this notion to public-key encryption of quantum data. A recent attack that exemplifies the relevance of the concept of non-malleability is the ``efail''-attack on the PGP protocol for confidential and authenticated e-mail communication \cite{efail}. 
This kind of attack, where an attacker is not directly able to learn the message yet still able to modify it, is exactly what non-malleable encryption secures against.

The classical notion of non-malleability is based on the notion of related plaintexts. For a non-malleable encryption scheme, it should, roughly speaking, be hard for an adversary to transform an encryption of a message $m$ into a \emph{different} ciphertext that decrypts to a \emph{related} message $m'$. Here, ``related'' just means that the adversary has some control over the transformation that is applied to the plaintext underlying the ciphertext he attacks. Generalizing this notion to the quantum case is complicated by the quantum no-cloning theorem: After a message has been encrypted and modified by the adversary and subsequently decrypted, it cannot be compared with the result anymore. In addition, it cannot be checked in a straightforward manner whether the adversary has indeed modified the ciphertext. 

In this work, we overcome these obstacles. The key ideas are the following. In the classical security game, an adversary is first asked to submit a distribution from which a plaintext is sampled. In the quantum setting, any message sampling procedure can be implemented by first performing a unitary quantum computation, and then discarding the contents of an auxiliary register. Instead of discarding this register, we view it as an extra record that is created along with the message. This extra record is then used instead of the original plaintext for evaluating the quantum analog of a relation. The test whether the adversary has indeed modified the ciphertext is performed by running the sampling- and encryption computations backwards on the attacked ciphertext. If the ciphertext was \emph{not} modified, this returns the registers into their initial blank state, which can be detected. 

We establish confidence in the new security notion by showing that it becomes equivalent to the classical notion when restricted to the post-quantum setting, i.e.\ to classical PKE schemes and classical plaintexts and ciphertexts. We also show how to satisfy the new security notion using a classical-quantum hybrid construction.

Along the way, we chart the landscape of one-time non-malleability notions for symmetric-key quantum encryption. We propose definitions for plaintext and ciphertext non-malleability and explore their relationship with existing definitions. In particular, we present evidence that these notions are the right ones.

\subsection{Related Work}

Non-malleability has been studied extensively in the classical setting, see \cite{bellare1999non,Pass2006} and references therein. In quantum cryptography, non-malleability has been, to our knowledge, subject of only two earlier works \cite{ambainis2009nonmalleable,NonMal}, which were only concerned with one-time security for symmetric-key encryption. 

Quantum public-key encryption has  been studied in \cite{Broadbent2015,Alagic2016} with respect to confidentiality.

Problems due to quantum no-cloning and the destructive nature of quantum measurement similar to the ones we face in this work have arisen before in the literature.  In particular, devising security notions for quantum encryption where the classical security definition requires copying and comparing plaintexts or ciphertexts \cite{alagic2018unforgeable,SignCrypt}, as well as in some quantum attack models for classical cryptography \cite{Boneh2013,Boneh2013a,Alagic2018} requires tackling similar obstacles. Another important case where the generalization of classical techniques is complicated by the mentioned features of quantum theory is that of rewinding and reprogramming \cite{Unruh2012, Watrous2018, Don2019}.

\subsection{Summary of Contributions}

The contributions presented in this paper can be divided into two categories, depending on whether they concern symmetric-key encryption (SKE) or public-key encryption (PKE). While we consider our results of the latter kind our main contribution, they build upon the former results. We therefore begin by presenting our results on one-time non-malleability of quantum SKE in Section~\ref{sec:SKE}, after which we continue with the results on many-time  non-malleability for quantum PKE in Section~\ref{sec:PKE}.

\subsubsection{Symmetric-Key Non-Malleability}

All security notions that are concerned with malleability attacks come in two flavors, a \emph{plaintext} and a \emph{ciphertext} one. The difference is that in the former, an attack that modifies a ciphertext into a different one that decrypts to the same plaintext is considered harmless, while the latter considers any modification yielding the encryption of a related plaintext a successful attack. We refine the non-malleability notion \NM introduced in \cite{NonMal}, to obtain a definition for both ciphertext and plaintext non-malleability, while staying in the \emph{effective-map based framework}. The effective map resulting from a one-time malleability attack is the map on the plaintext space, that is implemented by the procedure of encrypting the input, applying the attack and subsequently decrypting the result again.

\begin{definition}[\PNM and \CiNM, informal]
A scheme is \emph{plaintext non-malleable, (\PNM))}, if for any attack $\Lambda$ the effective map $\Tilde{\Lambda}$ consists of replacing the plaintext with a random decryption with some probability $p$, and leaving it unchanged otherwise. If the probability $p$ is equal to the probability that $\Lambda$ acts as the identity on a random ciphertext, it is even \emph{ciphertext non-malleable (\CiNM))}.
\end{definition}


There is one important subtlety that we would like to highlight here. The notion \NM from \cite{NonMal} is very similar to \CiNM. The only difference is that in the former, $p$ is derived from the probability that $\Lambda$ acts as identity on a random element of the ciphertext space, including ciphertexts that are not even valid. For \CiNM, on the other hand, $p$ is the probability that  $\Lambda$ acts as identity on a ciphertext that is generated by picking a random plaintext and then encrypting it.

We continue by exploring the relationship between \NM, \CiNM, and \PNM. In particular, we present separating examples between \NM and \CiNM, \NM and \PNM, and \CiNM and \PNM, and  show that both notions of ciphertext non-malleability, \NM and \CiNM, imply plaintext non-malleability,
\begin{theorem}[\ref{thm:NMorCiNMimpliesPNM}]\label{thm:NMorCiNMimpliesPNM-summary}
	Any $\eps$-\NM or $\eps$-\CiNM \SKQES is $\eps$-\PNM.
\end{theorem}
Intuitively, this result is because \PNM is obtained from \CiNM or \NM by removing the restraint on $p$. Additionally, we give a simplifying characterization of \PNM that allows for efficient simulation following \cite{Broadbent2016}.
We also show that for encryption schemes with unitary encryption map\footnote{More precisely, the encryption with a fixed key is unitary.}, all three notions are equivalent.
\begin{theorem}[\ref{thm:unitaryissimple}, informal]
For symmetric-key encryption schemes with unitary encryption map, \PNM, \CiNM, and \NM are equivalent. 
\end{theorem}

Finally, we show that one can construct a quantum authentication scheme according to the security definition from \cite{dupuis2012actively} from a \PNM scheme (and therefore, by  Theorem \ref{thm:NMorCiNMimpliesPNM-summary}, also from a \CiNM scheme). This is done by adding a tag to the plaintext during encryption, which is checked during decryption, as proposed for \NM schemes in \cite{NonMal}.

\begin{theorem}[\ref{thm:DNSfromPNM}]
From any \PNM scheme, one can construct a $2^{2-r}$-\DNS-authenticating scheme using $r$ tag qubits.
\end{theorem}

\subsubsection{Public-Key Non-Malleability}

We propose a definition for public-key quantum non-malleability in a computational setting, by adapting the classical definition for comparison-based non-malleability found in \cite{bellare1999non}, a real-vs-ideal definition. In the following, we describe informally what main challenges the generalization of the classical security experiments (the real and the ideal one) to the quantum setting poses, and how we resolved them. 

In the first step in the classical security experiments, the adversary submits a probability distribution $p$ over messages. In both experiments, a plaintext from this distribution is sampled, encrypted and sent to the adversary. The adversary now has the opportunity to manipulate (or ``malleate'') the ciphertext with the goal that the output decrypts to a \emph{related} plaintext.\footnote{In the actual experiments, the adversary is allowed to transform the ciphertext into many attempted manipulated ciphertexts. In this informal description we simplify as no significant additional technical challenges arise from the generalization.} The relation according to which the plaintexts are related, is supplied by the adversary. Of course there are examples of relations that allow for easy creation of a ciphertext that decrypts to a related plaintext, like e.g. the trivial relation where any plaintext is related to any other plaintext. To not credit an adversary with a break for fulfilling such a relation, her success in two experiments is compared: in the real world, the relation is evaluated on the initial and final plaintexts, but in the ideal world, it is evaluated on the final plaintext and a plaintext that is independently sampled from $p$.

Attempting a naive quantum generalization, we face two main challenges: How does the challenger ensure that the ciphertext he received from the adversary is actually modified? And how does he evaluate a relation on the input plaintext and the decrypted one? Both questions are complicated by the fact that quantum information cannot be copied. The first question has a rather elegant solution. Instead of asking the adversary to provide a distribution of messages, we ask her to provide a state preparation circuit, a strict generalization of the former. Such a state preparation circuit starts from a blank register and prepares a quantum state on the plaintext register and an auxiliary register. But quantum operations are reversible, which means that to test whether the plaintext has changed after encryption, attack and decryption,\footnote{Here we have to undo encryption in a different way as decryption is an irreversible process in general, see Section~\ref{sec:PKE} for details.} we can run preparation backwards and measure whether we got back a blank register. If so, the ciphertext was not changed, and the candidate manipulated plaintext is discarded. If not, we run preparation forward again, recovering (the actually changed part of) the adversary's candidate malleation.

The second question is solved by exploiting the fact that in the quantum setting, any message-sampling procedure can be implemented by first performing a unitary quantum computation, and then discarding the contents of an auxiliary register. We can therefore ask the adversary to provide a state-sampling unitary and store the auxiliary register as a record indicating which plaintext has been created. After proceeding with the experiment as in the classical case and using the modification test as described above, instead of checking whether the original and the attacked plaintext are related, we can now, in the real world case, check whether the attacked plaintext is related to the record. In the ideal world, the record is replaced with an independently created one.


\begin{definition}[\QCNM, informal]
	A scheme is quantum comparison-based non-malleable (\QCNM) if no adversary can achieve a better than negligible advantage in distinguishing the real and ideal versions of the quantum comparison-based non-malleability experiment described above.
\end{definition}
We go on to show that \QCNM is  a consistent generalization of \CNM.
\begin{theorem}[\ref{thm:QCNMequivPQ}, informal]
When restricted to the post-quantum setting, \QCNM and \CNM are equivalent.
\end{theorem}

Finally, we show that a \QCNM scheme can be constructed from a \CiNM scheme. The quantum-classical hybrid construction, which was extensively studied in \cite{SignCrypt} with respect to confidentiality and integrity,
 is obtained by encrypting every plaintext with a symmetric-key, one-time secure quantum encryption scheme and a fresh key, and then encrypting that key with a non-malleable classical public-key encryption scheme and appending it to the ciphertext.
\begin{theorem}[\ref{thm:security}, informal]
Using a \CNM classical scheme and a \CiNM quantum scheme it is possible to construct a \QCNM scheme via quantum-classical hybrid encryption.
\end{theorem}


\section{Preliminaries}
\label{sec:prelims}

In this section, we introduce the notation and conventions used and provide a very brief overview of background material. For a more general overview of quantum computing see, for example, \cite{Watrous2018}.

\subsection{Conventions and Notation}
The adjoint of a complex matrix $M$ is denoted by $M^\dagger$ and its trace as $\Trr{M}$. All Hilbert spaces $\HH_A$ in this work have dimension $|A|:= \dim(\HH_A) = 2^m$ for some $m\in\N$. For Hilbert spaces $\HH_A$, and $\HH_B$, we write $\I^{\ch{A}}$ for the identity matrix on $\HH_A$, or $\I$ if the space is clear from context, and $\Z^{\ch[A]{B}}$ or $\Z^{\ch{A}}$ for the all-zero matrix of dimension $|A| \times |B|$ or $|A| \times |A|$ respectively. 
We denote the set of square matrices that act on $\HH_A$ as $\MB(\HH_A)$. We call a function $\eps(n)$ negligible (denoted $\eps \leq \negl(n)$) if for every polynomial $p$ there exists $n_0 \in \N$ such that for all $n \geq n_0$ it holds that $\eps(n) < \frac{1}{p(n)}$. Furthermore we use $\log(x)$ to denote the base-$2$ logarithm of $x$.

\subsection{Quantum States and Operations}
We use bra-ket notation to denote a norm-$1$ vector $\ket{\phi} \in \HH_A$, sometimes denoted $\ket{\phi}^{\ch{A}}$ for clarity. The set $\{\ket{x}^{\ch{A}} \mid x \in \{0,1\}^n\}$ forms a basis of $\HH_A$ with $|A| = 2^n$, which is called the \emph{computational basis}. Quantum states are described by \emph{density matrices}, which are positive semi-definite Hermitian matrices with trace $1$. The set of density matrices on $\HH_A$ is denoted by $\MD(\HH_A)$. 
The \emph{maximally mixed state} is defined as $\tau^{\ch{A}} = \frac{\I}{|A|}$. Furthermore we use $\phi^{+\ch{AA'}} = \kb{\phi^+}{\phi^+}^{\ch{AA'}}$ to denote the (standard) \emph{maximally entangled state}, where $\ket{\phi^+}^{\ch{AA'}} = \frac{1}{\sqrt{|A|}}\sum\limits_{x\in \{0,1\}^{\log(|A|)}} \ket{xx}$.

A quantum state can be stored in a \emph{quantum register}, which can be thought of as the quantum equivalent of a variable. A register $A$ can store a density matrix $\rho \in \MD(\HH_A)$. In a cryptographic setting a ``register'' $A$ is often an infinite family of registers, one for each value of the security parameter. The action of a quantum algorithm can be described as a completely positive trace-preserving (\CPTP) map (a quantum channel). Sometimes the trace preserving property is relaxed to trace non-increasing, in which case we call it a \CPTNI-map. If a quantum algorithm has a classical argument then it is understood that this argument is converted to the computational basis and classical outputs are obtained by measuring in the computational basis. We write $\Lambda^{\ch[A]{B}}$ to mean a \CPTP map from register $A$ to register $B$. 
When a quantum channel $\Lambda^{\ch[A]{B}}$ is evaluated on a state $\rho^{\ch{AC}}$, then it implicitly acts as identity on register $C$, meaning $\Lambda^{\ch[A]{B}}(\rho^{\ch{AC}}) = (\Lambda^{\ch[A]{B}} \otimes \id^{\ch{C}})(\rho^{\ch{AC}})$. To quantify the difference between quantum channels we will use the diamond norm, or completely bounded trace norm, defined as
\[ \dn{L^{\ch[A]{B}}} = \max_{\rho^{\ch{AA'}}} \nm{(L \otimes \id^{\ch{A'}})(\rho)}_1,\]
where $A'$ is a copy of the $A$ register and $\nm{M}_1 = \Trr{\sqrt{M^\dagger M}}$.
For a quantum state $\sigma$, we define the CPTP map $\lr{\sigma}(\cdot)=\sigma\Tr(\cdot)$, i.e. $\lr{\sigma}$ is the constant quantum channel that maps every input state to $\sigma$.

We write $y \leftarrow A(x_1,\dots, x_n)$ to mean that $y$ is the result of running an algorithm $A$ on inputs $x_1,\dots, x_n$, and similarly $Y \leftarrow A(X_1,\dots, X_n)$ to mean that register $Y$ holds the state resulting from running the quantum algorithm $A$ on input registers $X_1,\dots,X_n$. We write PPT to denote a uniform polynomial-time family of classical circuits and QPT to denote a uniform polynomial-time family of quantum circuits.

\subsection{(Quantum) Encryption Schemes} We follow the conventions used in \cite{alagic2018unforgeable}, in particular we use $\Enc_k = \Enc(k, \cdot)$ and $\Dec_k = \Dec(k, \cdot)$. We begin by defining symmetric-key and public-key quantum encryption schemes.

\begin{definition}
A \emph{symmetric-key quantum encryption scheme (\SKQES)} is a triple $(\KeyGen, \Enc, \Dec)$, where
\begin{itemize}
    \item $\KeyGen$ is a PPT algorithm that given a security parameter $n\in\N$ outputs a key $k$,
    \item $\Enc$ is a QPT algorithm which takes as input a classical key $k$ and a quantum state in register $M$ and outputs a quantum state in register $C$,
    \item $\Dec$ is a QPT algorithm which takes as input a classical key $k$ and a quantum state in register $C$ and outputs a quantum state in register $M$ or $\kb{\bot}{\bot}^{\ch{\bot}}$,
\end{itemize}
such that $\dn{\Dec_k \circ \Enc_k - \id^{\ch[M]{M\oplus\bot}}} \leq \negl(n)$ for all $k \leftarrow \KeyGen(1^n)$.
\end{definition}

\begin{definition}
A \emph{public-key quantum encryption scheme (\PKQES)} is a triple $(\KeyGen, \Enc, \Dec)$, where
\begin{itemize}
    \item $\KeyGen$ is a PPT algorithm that given a security parameter $n \in \N$ outputs a pair of keys $(pk, sk)$,
    \item $\Enc$ is a QPT algorithm which takes as input a classical public key $pk$ and a quantum state in register $M$ and outputs a quantum state in register $C$,
    \item $\Dec$ is a QPT algorithm which takes as input a classical secret key $sk$ and a quantum state in register $C$ and outputs a quantum state in register $M$ or $\kb{\bot}{\bot}^{\ch{\bot}}$,
\end{itemize}
s.t. $\dn{\Dec_{sk} \circ \Enc_{pk} - \id^{\ch[M]{M\oplus\bot}}} \leq \negl(n)$ for all $(pk, sk) \leftarrow \KeyGen(1^n)$.
\end{definition}

It is implicit that 
$|M| \leq |C| \leq 2^{q(n)}$ for some polynomial $q$. Furthermore we only consider \emph{fixed-length} schemes, which means $|M|$ is a fixed function of $n$. Lastly we adopt the convention that every honest party applies the measurement $\{\kb{\bot}{\bot}, \I - \kb{\bot}{\bot}\}$ after running $\Dec$, and denote with $\Dec_k( C) \neq \bot$ the event that this measurement did not measure $\kb{\bot}{\bot}$ and thus produced a valid plaintext. Because of this convention we often state that the output space of $\Dec$ is $\MD(\HH_M)$ although it is technically $\MD(\HH_M \oplus \HH_\bot)$, where $\HH_\bot = \mathbb{C}\ket{\bot}$.

\begin{theorem}[Lemma 1 in \cite{SignCrypt}]
\label{thm:charac}\ \\
    Let $\Pi = (\KeyGen, \Enc,\Dec)$ be a \SKQES, then $\Enc$ and $\Dec$ have the following form, for all $k \leftarrow \KeyGen$:
    \begin{align*}
        &\dn{\Enc_k - V_k((\cdot)^{\ch{M}} \otimes \sigma_k^{\ch{T}})V_k^\dagger} \leq \eps\\
		&\dn{\Dec_k(V_k P_{\sigma_k}^{\ch{T}}(V_k^\dagger (\cdot)^{\ch{C}}V_k)P_{\sigma_k}^{\ch{T}}V_k^\dagger) - \Trr[T]{P_{\sigma_k}^{\ch{T}}(V_k^\dagger (\cdot)^{\ch{C}}V_k)P_{\sigma_k}^{\ch{T}}}} \leq \eps.
    \end{align*}
	Here $\sigma_k$ is a state on register $T$, $V_k$ is a unitary $\eps \leq \negl(n)$. Furthermore, $P_{\sigma_k}$ is an orthogonal projectors such that $\dn{P_{\sigma_k}\sigma_kP_{\sigma_k} - \sigma_k} \leq \eps$ and $\bar{P}_{\sigma_k} = \I - P_{\sigma_k}$. 
	
	Furthermore, for every $k$ there exists a probability distribution $p_k$ and a family of quantum states $\ket{\psi_{k,r}}^{\ch{T}}$ such that $\Enc_k$ is $\eps$-close to the following algorithm:
    \begin{enumerate}
        \item sample $r \xleftarrow{p_k} \{0,1\}^{\log |T|}$;
        \item apply the map $\Enc_{k;r}(X^{\ch{M}}) = V_k(X^{\ch{M}} \otimes \psi_{k,r}^{\ch{T}})V_k^\dagger$.
    \end{enumerate}
\end{theorem}

In this paper we will only consider schemes where all the actions described in Theorem \ref{thm:charac} can be implemented by a PPT or QPT algorithm.\footnote{Of course, the algorithms $\Enc $ and \Dec have some efficient implementation, but in principle that might differ from the above one.}

\subsection{Security Definitions}

In this paper we will build upon the classical definitions of non-malleability \cite{bellare1999non} and the existing quantum definitions of non-malleability \cite{ambainis2009nonmalleable,NonMal}.\\
\begin{minipage}{.46\textwidth}
\begin{algorithm}[H]
	\caption{\CNMreal}
\DontPrintSemicolon
\SetKwInOut{Input}{Input}\SetKwInOut{Output}{Output}
\Input{$\Pi, \MA, n$}
\Output{$b \in \{0,1\}$}
\BlankLine
$(pk, sk) \leftarrow \KeyGen(1^n)$\;
$(M,s) \leftarrow \MA_1(pk)$\;
$x \leftarrow M$\;
$y \leftarrow \Enc_{pk}(x)$\;
$(R, \mathbf{y}) \leftarrow \MA_2(s,y)$\;
$\mathbf{x} \leftarrow \Dec_{sk}(\mathbf{y})$\;
Output $1$ iff $(y\not\in\mathbf{y}) \land R(x,\mathbf{x})$
\end{algorithm}
\end{minipage}
\hspace{0.06\textwidth}
\begin{minipage}{.46\textwidth}
\begin{algorithm}[H]
	\caption{\CNMideal}
\DontPrintSemicolon
\SetKwInOut{Input}{Input}\SetKwInOut{Output}{Output}
\Input{$\Pi, \MA, n$}
\Output{$b \in \{0,1\}$}
\BlankLine
$(pk, sk) \leftarrow \KeyGen(1^n)$\;
$(M,s) \leftarrow \MA_1(pk)$\;
$x, \Tilde{x} \leftarrow M$\;
$\Tilde{y} \leftarrow \Enc_{pk}(\Tilde{x})$\;
$(R, \Tilde{\mathbf{y}}) \leftarrow \MA_2(s,\Tilde{y})$\;
$\Tilde{\mathbf{x}} \leftarrow \Dec_{sk}(\Tilde{\mathbf{y}})$\;
Output $1$ iff $(\Tilde{y}\not\in\Tilde{\mathbf{y}}) \land R(x,\Tilde{\mathbf{x}})$
\end{algorithm}
\end{minipage}

\begin{definition}[Definition 2 in \cite{bellare1999non} (CNM-CPA)]
A $\PKES$ $\Pi$ is \emph{com\-pa\-ri\-son-based non-malleable for chosen-plaintext attacks (\CNM)} if for any adversary $\MA = (\MA_1, \MA_2)$ it holds that
\[ \Prr{\CNMreal(\Pi, \MA, n) = 1} - \Prr{\CNMideal(\Pi, \MA, n) = 1} \leq \negl(n),\]
if $\MA$ is such that:
\begin{itemize}
    \item $\MA_1$ and $\MA_2$ are PPT
    \item $\MA_1$ outputs a valid message space $M$ which can be sampled by a PPT algorithm
    \item $\MA_2$ outputs a relation $R$ computable by a PPT algorithm
    \item $\MA_2$ outputs a vector $\mathbf{y}$ such that $\bot \not\in \Dec_{sk}(\mathbf{y})$
\end{itemize}
\end{definition}

For comparison-based non-malleability, we consider adversaries that are split into two stages, where each stage is a probabilistic algorithm. The first stage takes as input the public key and produces a message distribution, which is (a description of) a probabilistic algorithm that produces a plaintext. The second stage takes as input one ciphertext of a plaintext produced by this algorithm and produces a vector of ciphertexts and a relation $R$. The goal of the adversary is to construct $R$ in such a way that $R$ holds between the original plaintext and the (element-wise) decryption of the produced ciphertext vector, but not between another plaintext which is sampled independently from the message distribution and the decryption of this same vector. If an adversary can achieve this relation to hold with non-negligible probability, then intuitively the adversary was able to structurally change an encrypted message, which would indicate that the scheme is malleable.

In the existing literature on non-malleability in the quantum setting, the approach taken is quite different from the notion described above. Here, the focus is put on unconditional one-time security notions of symmetric-key non-malleability and authentication. In this setting, a notion of non-malleability was first introduced in \cite{ambainis2009nonmalleable}, which defines non-malleability as a condition on the effective map of an arbitrary attack. 
The \emph{effective map} of an attack $\Lambda_A^{\ch[CB]{C\hat{B}}}$ is defined as $\Tilde{\Lambda}_A^{\ch[MB]{M\hat{B}}} = \E\limits_{k\from\KeyGen(1^n)}\Dec_k \circ \Lambda^A \circ \Enc_k$, and can be thought of as the average effect of an attack on the plaintext level.

The main idea of this definition is  that a ciphertext cannot be meaningfully transformed into the ciphertext of another message, which means that the effective map of any attack is either identity, in case no transformation is applied, or a map $\lr{\rho}$, that replaces the ciphertext by a fixed one. Note that this way of defining non-malleability can also be satisfied by a scheme which has the property that an attacker can transform a ciphertext into another ciphertext of the same message. In other words, the non-malleability is only enforced on the plaintext level, which means it is a form of \emph{plaintext non-malleability}. The classical notions discussed in the previous section do not allow for attacks that map an encrypted message to a different encryption of the same message. This restriction means non-malleability is enforced on the ciphertext level and thus these classical notions define forms of \emph{ciphertext non-malleability}.

This effective-map-based way of describing non-malleability was continued in \cite{NonMal}, where an insufficiency of the previous definition was demonstrated and a new definition was given. Their definition is given in terms of the mutual information between the plaintext and the side-information collected by the attacker. However, one of the results in their paper is a characterization theorem which we consider as the definition instead.

\begin{definition}[Theorem 4.4 in \cite{NonMal}]
\label{def:nm}
A \SKQES\ $(\KeyGen, \Enc, \Dec)$ is \emph{$\eps$-non-malleable ($\eps$-\NM)} if, for any attack $\Lambda_A^{\ch[CB]{C\hat{B}}}$, its effective map\\ $\Tilde{\Lambda}_A^{\ch[MB]{M\hat{B}}}$ is such that
\[ \dn{\Tilde{\Lambda}_A - \left(\id^{\ch{M}} \otimes \Lambda_1^{\ch[B]{\hat{B}}} + \frac{1}{|C|^2 - 1}\left(|C|^2\lr{\Dec_K(\tau^{\ch{C}})} - \id\right)^{\ch{M}} \otimes \Lambda_2^{\ch[B]{\hat{B}}}\right)} \leq \eps,\]
where
\begin{align*}
    \Lambda_1 &= \Trr[CC']{\phi^{+\ch{CC'}}\Lambda_A(\phi^{+\ch{CC'}} \otimes (\cdot))} \qquad \text{and}\\
    \Lambda_2 &= \Trr[CC']{(\I^{\ch{CC'}} - \phi^{+\ch{CC'}})\Lambda_A(\phi^{+\ch{CC'}} \otimes (\cdot))}.
\end{align*}
A \SKQES is \emph{non-malleable (\NM)} if it is $\eps$-\NM for some $\eps \leq \negl(n)$.
\end{definition}

In the symmetric-key setting, one can also consider the notion of authentication. A scheme satisfying this notion not only prevents an attacker from meaningfully transforming ciphertexts, but any attempt to do so can also be detected by the receiving party. In \cite{dupuis2012actively} a definition is given for this notion, which we adapt slightly to use the diamond norm instead of the trace norm.

\begin{definition}[Definition 2.2 in \cite{dupuis2012actively}] \label{def:dns}
A \SKQES\ $\Pi$ is \emph{$\eps$-DNS authenticating ($\eps$-\DNS)} if, for any attack $\Lambda_A^{\ch[CB]{C\hat{B}}}$, its effective map $\Tilde{\Lambda}_A^{\ch[MB]{M\hat{B}}}$ is such that
 \[ \dn{\Tilde{\Lambda}_A - \left(\id^{\ch{M}} \otimes \Lambda_{acc}^{\ch[B]{\hat{B}}} + \lr{\kb{\bot}{\bot}} \otimes \Lambda_{rej}^{\ch[B]{\hat{B}}}\right)} \leq \eps,\]
for some \CPTNI\ maps $\Lambda_{acc}, \Lambda_{rej}$ such that $\Lambda_{acc} + \Lambda_{rej}$ is \CPTP.
A \SKQES\ $\Pi = (\KeyGen, \Enc, \Dec)$ is \emph{DNS authenticating (\DNS)} if it is $\eps-\DNS$ for some $\eps \leq \negl(n)$.
\end{definition}

It is shown in \cite{NonMal} that a \NM\ scheme can be modified to a scheme that is \DNS\ authenticating by appending a tag to the encoded plaintext.


\section{Non-Malleability for Quantum SKE} \label{sec:SKE}
While Definition~\ref{def:nm} of \NM presented in \cite{NonMal} has many desirable features, it turns out that it is slightly too strong in the sense that it rules out schemes that are clearly  non-malleable intuitively. Furthermore, it has not been discussed in \cite{NonMal} whether \NM actually ensures non-malleability of ciphertexts, or merely plaintext non-malleability. In this section, we will discuss these features of \NM in detail. Furthermore, we propose a plaintext and a ciphertext version of \NM, shedding light on how these different security properties are expressed in the effective-map formalism.

\subsection{Ciphertext Non-Malleability}

When inspecting the Definition~\ref{def:nm} of \NM, one can observe that the constraints on $\Lambda_1$ and $\Lambda_2$ make \NM a type of ciphertext non-malleability: Unless the adversary applies the identity channel, we end up in the case of $\Lambda_2$. However, the use of $\phi^+$ in defining these constraints can be considered problematic when the ciphertext space is not uniformly used, i.e. when $\Enc_K(\tau^{\ch{M}}) \neq \tau^{\ch{C}}$. We provide an example of how this could be problematic.


\begin{example}
Let $\Pi' = (\KeyGen, \Enc', \Dec')$ be an \NM \SKQES, with ciphertext space $\HH_{C'}$. Let $\HH_C = \HH_{C'} \otimes \HH_T$, where $\HH_T = \C^2$, then define $\Pi = (\KeyGen, \Enc, \Dec)$ as follows, with ciphertext space $\HH_C$:
\begin{itemize}
    \item $\Enc_k(X) = \Enc'_k(X) \otimes \kb{0}{0}^{\ch{T}}$
    \item $\Dec_k(Y) = \Dec'_k(\Trr[T]{\kb{0}{0}^{\ch{T}}Y}) + \Trr{\kb{1}{1}^{\ch{T}}Y}\kb{\bot}{\bot}$
\end{itemize}

Consider the attack $\Lambda(\psi^{\ch{C}}) = \kb{0}{0}^{\ch{T}}\psi\kb{0}{0}^{\ch{T}} + \lr{\tau}^{\ch{C'}}(\kb{1}{1}^{\ch{T}}\psi\kb{1}{1}^{\ch{T}})$ (with trivial register $B$), which is the attack of measuring the $T$ register in the computational basis and replacing the $C'$ register with the maximally mixed state if the outcome of this measurement is $1$ and doing nothing otherwise. As the register $B$ is trivial, $\Lambda_1$ is just a probability. We calculate
\begin{align*}
    \Lambda_1 &= \Trr{\phi^{+\ch{CC}}\Lambda(\phi^{+\ch{CC}})} \\
    &= \Trr{\phi^{+\ch{CC}}(\frac{1}{2}\kb{0}{0} \otimes \kb{0}{0} \otimes \phi^{+\ch{C'C'}} + \frac{1}{2}\kb{1}{1}\otimes\kb{1}{1}\otimes\tau^{\ch{C'}} \otimes \tau^{\ch{C'}})}\\
    &= \frac{|C'|^2 + 1}{4|C'|^2}.
\end{align*}
However the effective map is $\Tilde{\Lambda} = \id$, which shows that $\Pi$ is not \NM.
\end{example}

What could be considered problematic about this example is that any attack on $\Pi$ is also an attack on $\Pi'$, since the attacker could add and remove the $T$ register himself. Furthermore, there is a one-to-one correspondence between ciphertexts of $\Pi$ and ciphertexts of $\Pi'$, because the $T$ register is checked during decryption. This means that if an attacker could perform a malleability attack on $\Pi$, i.e. constructively transform a ciphertext into another ciphertext, then the attack obtained by applying the above strategy would be a malleability attack on $\Pi'$. Thus one could argue that, intuitively, non-malleability of $\Pi'$ should imply non-malleability of $\Pi$. We suggest the following improved definition that prevents this behavior.

\begin{definition}
\label{def:cinm}
A \SKQES\ $(\KeyGen, \Enc, \Dec)$ is \emph{$\eps$-ciphertext non-malleable ($\eps$-\CiNM)} if, for any attack $\Lambda_A^{\ch[CB]{C\hat{B}}}$, its effective map $\Tilde{\Lambda}_A^{\ch[MB]{M\hat{B}}}$ is such that
\[ \dn{\Tilde{\Lambda}_A - \left(\id^{\ch{M}} \otimes \Lambda_1^{\ch[B]{\hat{B}}} + \frac{1}{|C|^2 - 1}\left(|C|^2\lr{\Dec_K(\tau^{\ch{C}})} - \id\right)^{\ch{M}} \otimes \Lambda_2^{\ch[B]{\hat{B}}}\right)} \leq \eps,\]
where
\begin{align*}
    \Lambda_1 &= \E\limits_{k,r}\left[\Trr[CM']{\psi^{\ch{CM'}}_{k,r}\Lambda_A(\psi^{\ch{CM'}}_{k,r} \otimes (\cdot))}\right]&and\\
    \Lambda_2 &= \E\limits_{k,r}\left[\Trr[CM']{(\I^{\ch{CM'}} - \psi^{\ch{CM'}}_{k,r})\Lambda_A(\psi^{\ch{CM'}}_{k,r} \otimes (\cdot))}\right].
\end{align*}
Here $\Enc_{k;r}$ is as in Theorem \ref{thm:charac}, $\E_{k,r}$ is taken uniformly over $k$ and with $r$ sampled according to $p_k$ from Theorem \ref{thm:charac}, and $\psi^{\ch{CM'}}_{k,r} = \Enc_{k;r}(\phi^{+\ch{MM'}})$.
A \SKQES is \emph{ciphertext non-malleable (\CiNM)} if it is $\eps$-\CiNM for some $\eps \leq \negl(n)$.
\end{definition}

\subsection{Plaintext Non-Malleability}

For ciphertext non-malleability, discussed in the last section, the effective map approach seems slightly ill-suited: after all, the effective map is a map on plaintexts! What makes \CiNM (and \NM, albeit in an overzealous way) definitions of ciphertext non-malleability are the constraints placed on the map which $\Tilde{\Lambda}_A$ is compared with (the \emph{simulator}). These constraints are imposed by the definitions of $\Lambda_1$ and $\Lambda_2$ and connect the simulator, which acts on plaintexts, to the attack map, which acts on ciphertexts.
In order to construct a definition for plaintext non-malleability from \NM, we therefore drop these constraints. In addition, we change the $|C|^2$ constant for the constant $|M|^2$, as the former constant is a direct artifact of the constraints. In other words, plaintext-non-malleability ``does not know about ciphertexts'', i.e., in particular, the ciphertext space dimension should be immaterial. We would like to remark that the latter point does not matter when talking about approximate non-malleability in the asymptotic setting, where the plaintext space grows polynomially with the security parameter.

The above considerations lead to the following definition.

\begin{definition}
\label{def:pnm}

A \SKQES\ $(\KeyGen, \Enc, \Dec)$ is \emph{$\eps$-plaintext non-malleable ($\eps$-\PNM)} if, for any attack $\Lambda_A^{\ch[CB]{C\hat{B}}}$, its effective map $\Tilde{\Lambda}_A^{\ch[MB]{M\hat{B}}}$ is such that
\[ \dn{\Tilde{\Lambda}_A - \left(\id^{\ch{M}} \otimes \Lambda_1^{\ch[B]{\hat{B}}} + \frac{1}{|M|^2 - 1}\left(|M|^2\lr{\Dec_K(\tau^{\ch{C}})} - \id\right)^{\ch{M}} \otimes \Lambda_2^{\ch[B]{\hat{B}}}\right)} \leq \eps,\]
where $\Lambda_1$ and $\Lambda_2$ are \CPTNI\ and $\Lambda_1 + \Lambda_2$ is \CPTP.
A \SKQES is \emph{plaintext non-malleable (\PNM)} if it is $\eps-\PNM$ for some $\eps \leq \negl(n)$.
\end{definition}

Intuitively, ciphertext non-malleability is a strictly stronger security notion than plaintext non-malleability since the latter is obtained from the former by dropping the constraints on the simulator. This intuition holds true for our proposed \PNM definition.

\begin{lemma}
\label{lem:CtoM}
Let $\Pi = (\KeyGen, \Enc, \Dec)$ be an arbitrary \SKQES\ and $\Lambda_A^{\ch[CB]{C\hat{B}}}$ an arbitrary attack on $\Pi$ with effective map $\Tilde{\Lambda}_A^{\ch[MB]{M\hat{B}}}$. If there exist \CPTNI\ $\Lambda_1, \Lambda_2$, such that $\Lambda_1 + \Lambda_2$ is \CPTP and it holds that
\[ \dn{\Tilde{\Lambda}_A - \left(\id^{\ch{M}} \otimes \Lambda_1^{\ch[B]{\hat{B}}} + \frac{1}{|C|^2 - 1}\left(|C|^2\lr{\Dec_K(\tau^{\ch{C}})} - \id\right)^{\ch{M}} \otimes \Lambda_2^{\ch[B]{\hat{B}}}\right)} \leq \eps,\]
then for any $\alpha$ such that $|M|^2 \leq \alpha \leq |C|^2$ there exist \CPTNI\ $\Lambda_3, \Lambda_4$ such that $\Lambda_3 + \Lambda_4$ is \CPTP\ and
\[ \dn{\Tilde{\Lambda}_A - \left(\id^{\ch{M}} \otimes \Lambda_3^{\ch[B]{\hat{B}}} + \frac{1}{\alpha - 1}\left(\alpha\lr{\Dec_K(\tau^{\ch{C}})} - \id\right)^{\ch{M}} \otimes \Lambda_4^{\ch[B]{\hat{B}}}\right)} \leq \eps.\]
\end{lemma}
\begin{proof}
For fixed $\Lambda_1$ and $\Lambda_2$ one can obtain the statement by defining
\begin{align*}
    \Lambda_3 &= \Lambda_1 + \left(1 - \frac{(\alpha - 1)|C|^2}{\alpha(|C|^2-1)}\right)\Lambda_2&\text{and}\\
    \Lambda_4 &= \frac{(\alpha - 1)|C|^2}{\alpha(|C|^2-1)}\Lambda_2.
\end{align*}
The full proof of this lemma is rather technical and can be found in Appendix \ref{lem:prf_CtoM}.
\end{proof}

Lemma \ref{lem:CtoM} shows that the $|C|^2$ constant present in the \NM definition can be decreased down to $|M|^2$, obtaining increasingly weaker security notions. This fact immediately implies the following

\begin{theorem}\label{thm:NMorCiNMimpliesPNM}
Any $\eps$-\NM or $\eps$-\CiNM \SKQES is $\eps$-\PNM.
\end{theorem}
\begin{proof}
This follows directly from Lemma \ref{lem:CtoM} with $\alpha = |M|^2$.
\end{proof}

While \PNM does not explicitly restrict the choice of $\Lambda_1$ and $\Lambda_2$, an explicit form for $\Lambda_i$ can be required without significantly strengthening the definition in the sense that the additional requirement only decreases security by at most a factor of 3.

\begin{theorem}
\label{thm:pnm_charac}
Let $\Pi = (\KeyGen, \Enc, \Dec)$ be an arbitrary $\eps$-\PNM\ \SKQES\ for some $\eps$, then for any attack $\Lambda_A^{\ch[CB]{C\hat{B}}}$, its effective map $\Tilde{\Lambda}_A^{\ch[MB]{M\hat{B}}}$ is such that
\[ \dn{\Tilde{\Lambda}_A - \left(\id^{\ch{M}} \otimes \Lambda_1^{\ch[B]{\hat{B}}} + \frac{1}{|M|^2 - 1}\left(|M|^2\lr{\Dec_K(\tau^{\ch{C}})} - \id\right)^{\ch{M}} \otimes \Lambda_2^{\ch[B]{\hat{B}}}\right)} \leq 3\eps,\]
where
\begin{align*}
    \Lambda_1 &= \Trr[MM']{\phi^{+\ch{MM'}}\Tilde{\Lambda}_A(\phi^{+\ch{MM'}} \otimes (\cdot))}&and\\
    \Lambda_2 &= \Trr[MM']{(\I^{\ch{MM'}} - \phi^{+\ch{MM'}})\Tilde{\Lambda}_A(\phi^{+\ch{MM'}} \otimes (\cdot))}.
\end{align*}
\end{theorem}
\begin{proof}
We sketch the proof here, the full proof of this theorem can be found in Appendix \ref{thm:prf_pnm_charac}.
Let $\Pi = (\KeyGen, \Enc, \Dec)$ be an arbitrary $\eps$-\PNM\ \SKQES\ for some $\eps$ and let $\Lambda_A^{\ch[CB]{C\hat{B}}}$ be an arbitrary attack with effective map $\Tilde{\Lambda}_A^{\ch[MB]{M\hat{B}}}$. Furthermore, let $\Lambda_1^{\ch[B]{\hat{B}}}$ and $\Lambda_2^{\ch[B]{\hat{B}}}$ be such that $\dn{\Tilde{\Lambda}_A - \Tilde{\Lambda}_{ideal}} \leq \eps$,
where $\Tilde{\Lambda}_{ideal}^{\ch[MB]{M\hat{B}}} = \id^{\ch{M}} \otimes \Lambda_1 + \frac{1}{|M|^2-1}(|M|^2\lr{\Dec_K(\tau)} - \id)^{\ch{M}} \otimes \Lambda_2$. Lastly, let
\begin{align*}
    \Lambda_3 &= \Trr[MM']{\phi^{+\ch{MM'}}\Tilde{\Lambda}_A(\phi^{+\ch{MM'}} \otimes (\cdot))},\\
    \Lambda_4 &= \Trr[MM']{(\I^{\ch{MM'}} - \phi^{+\ch{MM'}})\Tilde{\Lambda}_A(\phi^{+\ch{MM'}} \otimes (\cdot))},\\
    \Tilde{\Lambda}_{trace}^{\ch[MB]{M\hat{B}}} &= \id^{\ch{M}} \otimes \Lambda_3 + \frac{1}{|M|^2-1}(|M|^2\lr{\Dec_K(\tau)} - \id)^{\ch{M}} \otimes \Lambda_4,\\
    \Lambda_5 &= \Trr[MM']{\phi^{+\ch{MM'}}\Tilde{\Lambda}_{ideal}(\phi^{+\ch{MM'}} \otimes (\cdot))}\text{, and}\\
    \Lambda_6 &= \Trr[MM']{(\I^{\ch{MM'}} - \phi^{+\ch{MM'}})\Tilde{\Lambda}_{ideal}(\phi^{+\ch{MM'}} \otimes (\cdot))}
\end{align*}
Observe that $\dn{\Tilde{\Lambda}_{ideal} - \Tilde{\Lambda}_{trace}} \leq \dn{ (\Lambda_1 - \Lambda_3)} + \dn{ (\Lambda_2 - \Lambda_4)}$. Since\\ $\dn{(\Tilde{\Lambda} - \Tilde{\Lambda}_{ideal})(\phi^{+\ch{MM'}} \otimes (\cdot))} \leq \dn{\Tilde{\Lambda} - \Tilde{\Lambda}_{ideal} } \leq \eps$, we have $\dn{\Lambda_3 - \Lambda_5} \leq \eps$ and $\dn{\Lambda_4 - \Lambda_6} \leq \eps$. Using this we observe that
\begin{align*}
    \dn{\Tilde{\Lambda}_{ideal} - \Tilde{\Lambda}_{trace} } &\leq\dn{ \Lambda_1 - \Lambda_3} + \dn{ \Lambda_2 - \Lambda_4}\\
    &\leq \dn{ \Lambda_1 - \Lambda_5} + \dn{ \Lambda_5 - \Lambda_3} +\dn{ \Lambda_2 - \Lambda_6} + \dn{ \Lambda_6 - \Lambda_4}\\
    &\leq 2\eps + \dn{ \Lambda_1 - \Lambda_5} + \dn{ \Lambda_2 - \Lambda_6}.
\end{align*}

By substituting the definition of $\Tilde{\Lambda}_{ideal}$ we observe that $\Lambda_5 = \Lambda_1$ and $\Lambda_6 = \Lambda_2$. From this we conclude
\begin{align*}
     \dn{\Tilde{\Lambda} - \Tilde{\Lambda}_{trace} } &\leq \dn{\Tilde{\Lambda} - \Tilde{\Lambda}_{ideal} } + \dn{\Tilde{\Lambda}_{ideal} - \Tilde{\Lambda}_{trace} }\\
     &\leq 3\eps.
\end{align*}
\end{proof}

\begin{theorem}
\label{thm:pnm_not_nm}
There exists a \PKQES\ $\Pi = (\KeyGen, \Enc, \Dec)$ that is \PNM\ but not \NM and not \CiNM.
\end{theorem}
\begin{proof}
Let $\Pi' = (\KeyGen', \Enc', \Dec')$ be an arbitrary \PKQES\ that is \NM\footnote{See \cite{NonMal} for such a scheme.}. Then define $\Pi = (\KeyGen, \Enc, \Dec)$ as
\begin{align*}
    \KeyGen &= \KeyGen'\\
    \Enc_k &= \Enc_k' \otimes \kb{0}{0}^{\ch{R}}\\
    \Dec_k &= \Dec_k' \circ \Tr_R,
\end{align*}
where $R$ is an auxiliary 1-qubit register. Let $\Lambda$ be an arbitrary attack on $\Pi$ with effective map $\Tilde{\Lambda}$, then define $\Lambda' = \Trr[R]{\Lambda((\cdot) \otimes \kb{0}{0}^{\ch{R}})}$, which is an attack on $\Pi'$ with effective map $\Tilde{\Lambda}'$. Observe that $\Tilde{\Lambda}' = \Tilde{\Lambda}$, since the $R$ register is only added and then traced out.
Because $\Pi'$ is \NM, we have $\Lambda_3^{\ch[B]{\hat{B}}}$ and $\Lambda_4^{\ch[B]{\hat{B}}}$ such that 
\begin{align*}
  \dn{\Tilde{\Lambda}' - \id^{\ch{M}} \otimes \Lambda_3 + \frac{1}{|C|^2 - 1}\left(|C|^2\lr{\Dec_K(\tau^{\ch{C}})} - \id\right)^{\ch{M}} \otimes \Lambda_4} \leq \negl(n)
\end{align*}

It follows from Lemma \ref{lem:CtoM} that $\Pi$ is \PNM.

Now consider the attack $\Lambda_X = \id^{\ch{C}} \otimes (X(\cdot)X)^{\ch{R}} \otimes \Tr[\cdot]^{\ch{B}}$, where $X$ is the Pauli $X$ gate, with $X\ket{0} = \ket{1}$ and $X\ket{1} = \ket{0}$. Let $f(x_1\dots x_n) = x_1\dots x_{n-1}(1-x_n)$, i.e. the result of flipping the last bit of some bitstring. Observe that
\begin{align*}
    \Lambda_X(\phi^{+\ch{CC'}} \otimes (\cdot)^{\ch{B}}) &= (\id \otimes (X(\cdot)X)^{\ch{R}}\otimes \Tr[\cdot]^{\ch{B}})\left(\sum\limits_{i,j}\kb{ii}{jj}^{\ch{CC'}} \otimes (\cdot)^{\ch{B}}\right)\\
    &= \Trr{\cdot}^{\ch{B}}\sum\limits_{i,j}\kb{f(i)i}{f(j)j}.
\end{align*}
Since this superposition contains no components of the form $\kb{xx}{xx}^{\ch{CC'}}$ and $\phi^{+\ch{CC'}}$ only contains components of this form, we have that $\phi^{+\ch{CC'}}\Lambda_X(\phi^{+\ch{CC'}}\otimes(\cdot)^{\ch{B}}) = \Z^{\ch[BCC']{CC'}}$. With $\psi^{\ch{CM'}}_{k,r} = \Enc_{k;r}(\phi^{+\ch{MM'}}) = \kb{0}{0} \otimes \Enc'_{k;r}(\phi^{+\ch{MM'}})$, we have $\psi^{\ch{CM'}}_{k,r}\Lambda_X(\psi^{\ch{CM'}}_{k,r}\otimes(\cdot)^{\ch{B}}) = \Z^{\ch[BCC']{CC'}}$. 

Also note that the effective map of $\Lambda_X$ is $\Tilde{\Lambda}_X = \id^{\ch{M}} \otimes \Tr[\cdot]^{\ch{B}}$, since the attack only acts on $R$ and $B$ and thus does not modify the message in $M$. Let $\Lambda_1$ and $\Lambda_2$ be as in Definition \ref{def:nm} or as in Definition \ref{def:cinm}, then $\Trr{\Lambda_1(\rho)} = 0$ for all $\rho$. It follows that $\Lambda_2 = \Trr[CC']{(\I - \phi^{+\ch{CC'}})\Lambda_X(\phi^{+\ch{CC'}}\otimes(\cdot)^{\ch{B}})} = \Tr[\cdot]^{\ch{B}}$. Furthermore we have
\begin{align*}
    &\dn{\Tilde{\Lambda}_X - \left(\id^{\ch{M}} \otimes \Lambda_1 + \frac{1}{|C|^2 - 1}\left(|C|^2\lr{\Dec_K(\tau^{\ch{C}})} - \id\right)^{\ch{M}} \otimes \Lambda_2\right)}\\
    = &\dn{\left(\id^{\ch{M}} \otimes \Tr[\cdot]^{\ch{B}}\right) - \left(\frac{1}{|C|^2 - 1}\left(|C|^2\lr{\Dec_K(\tau^{\ch{C}})} - \id\right)^{\ch{M}} \otimes \Tr[\cdot]^{\ch{B}}\right)}\\
    = &\dn{\id^{\ch{M}} - \frac{1}{|C|^2 - 1}\left(|C|^2\lr{\Dec_K(\tau)} - \id\right)^{\ch{M}}}\\
    \geq &\nm{\phi^{+\ch{MM'}} - \frac{1}{|C|^2 - 1}(|C|^2\Dec_K(\tau)\otimes\tau^{\ch{M'}} - \phi^{+\ch{MM'}})}_1\\
    = &2\max\limits_{0 \leq P \leq \I} \Trr{P(\phi^{+\ch{MM'}} - \frac{1}{|C|^2 - 1}(|C|^2\Dec_K(\tau)\otimes\tau^{\ch{M'}} - \phi^{+\ch{MM'}}))}\\
    \geq &2\Trr{\phi^{+\ch{MM'}}(\phi^{+\ch{MM'}} - \frac{1}{|C|^2 - 1}(|C|^2\Dec_K(\tau)\otimes\tau^{\ch{M'}} - \phi^{+\ch{MM'}}))}\\
    = &2 - \frac{2(|C|^2 - |M|^2)}{|M|^2(|C|^2-1)} > 1,
\end{align*}
where we use that $\Trr{\phi^{+\ch{MM'}}(\Dec_K(\tau)\otimes\tau^{\ch{M'}})} = \frac{1}{|M|^2}$, as is proven in the proof of Theorem \ref{thm:prf_pnm_charac}, and $|M| \geq 2$, which is true when we assume that we are encrypting at least one qubit. This shows that $\Pi$ is not \NM and not \CiNM.
\end{proof}

While the above shows that \PNM and \CiNM are not the same in general, a special case arises when each plaintext has exactly one ciphertext (per key). Recall that plaintext non-malleability relaxes the constraints of ciphertext non-malleability by allowing the adversary to implement an attack that transforms one ciphertext into another, as long as both decrypt to the same plaintext. Thus in this special case, this relaxation is no relaxation at all. This special case arises in particular when an encryption scheme is \emph{unitary}, meaning that $\Enc_k(X) = V_kXV_k^\dagger$ for some collection $\{V_k\}_k$ of unitaries $V_k\in\mathrm U(\HH_M)$.

\begin{theorem}\label{thm:unitaryissimple}
For any unitary \SKQES $\Pi$, $\Pi$ is \PNM iff $\Pi$ is \CiNM iff $\Pi$ is \NM. 
\end{theorem}
\begin{proof}
Since $\CiNM, \NM \Rightarrow \PNM$ for all \SKQES, we only need to show the converse direction. Let $\Pi = (\KeyGen, \Enc, \Dec)$ be a \PNM unitary \SKQES and $\Lambda_A$ an arbitrary attack on this scheme. By Theorem \ref{thm:pnm_charac}, we have, for some $\eps \leq \negl(n)$, that 
\[ \dn{\Tilde{\Lambda}_A - \left(\id^{\ch{M}} \otimes \Lambda_1^{\ch[B]{\hat{B}}} + \frac{1}{|M|^2 - 1}\left(|M|^2\lr{\Dec_K(\tau^{\ch{C}})} - \id\right)^{\ch{M}} \otimes \Lambda_2^{\ch[B]{\hat{B}}}\right)} \leq 3\eps,\]
where
\begin{align*}
    \Lambda_1 &= \Trr[MM']{\phi^{+\ch{MM'}}\Tilde{\Lambda}_A(\phi^{+\ch{MM'}} \otimes (\cdot))}&and\\
    \Lambda_2 &= \Trr[MM']{(\I^{\ch{MM'}} - \phi^{+\ch{MM'}})\Tilde{\Lambda}_A(\phi^{+\ch{MM'}} \otimes (\cdot))}.
\end{align*}
Let $\{V_k^{\ch{M}}\}_k$ be the collection such that $\Enc_k(X) = V_kXV_k^\dagger$ and note that $\Enc_{k;r} = \Enc_k$ and $C = M$, where $\Enc_{k;r}$ is as in Theorem \ref{thm:charac}. Observe that
\begin{align*}
    \Lambda_1 &= \Trr[MM']{\phi^{+\ch{MM'}}\E\limits_k[\Dec_k(\Lambda_A(\Enc_k(\phi^{+\ch{MM'}} \otimes (\cdot))))]}\\
    &= \Trr[MM']{\phi^{+\ch{MM'}}\E\limits_k\left[V_k^\dagger(\Lambda_A(V_k(\phi^{+\ch{MM'}} \otimes (\cdot))V_k^\dagger))V_k\right]}\\
    &= \E\limits_k\left[\Trr[MM']{V_k\phi^{+\ch{MM'}}V_k^\dagger(\Lambda_A(V_k(\phi^{+\ch{MM'}} \otimes (\cdot))V_k^\dagger))}\right]\\
    &= \E\limits_{k,r}\left[\Trr[CM']{\psi^{\ch{CM'}}_{k,r}\Lambda_A(\psi^{\ch{CM'}}_{k,r} \otimes (\cdot))}\right],
\end{align*}
where $\psi_{k,r} = V_k\phi^{+\ch{MM'}}V_k^\dagger = \Enc_{k;r}(\phi^{+\ch{MM'}})$. In the same way one can deduce that $\Lambda_2 = \E\limits_{k,r}\left[\Trr[CM']{(\I^{\ch{CM'}} - \psi^{\ch{CM'}}_{k,r})\Lambda_A(\psi^{\ch{CM'}}_{k,r} \otimes (\cdot))}\right]$, and thus $\Pi$ is \CiNM.
Similarly, we have
\begin{align*}
    \Lambda_1 &= \Trr[MM']{\phi^{+\ch{MM'}}\E\limits_k\left[V_k^\dagger(\Lambda_A(V_k(\phi^{+\ch{MM'}} \otimes (\cdot))V_k^\dagger))V_k\right]}\\
    &= \Trr[MM']{\phi^{+\ch{MM'}}\E\limits_k\left[V_k^\dagger(\Lambda_A(V_k^{T\ch{M'}}(\phi^{+\ch{MM'}} \otimes (\cdot))\bar{V_k}^{\ch{M'}}))V_k\right]}\\
    &= \Trr[MM']{\phi^{+\ch{MM'}}\E\limits_k\left[(V_k^\dagger \otimes V_k^{T\ch{M'}})(\Lambda_A(\phi^{+\ch{MM'}} \otimes (\cdot)))(V_k\otimes\bar{V_k}^{\ch{M'}})\right]}\\
    &= \E\limits_k\left[\Trr[MM']{(V_k\otimes\bar{V_k}^{\ch{M'}})\phi^{+\ch{MM'}}(V_k^\dagger \otimes V_k^{T\ch{M'}})(\Lambda_A(\phi^{+\ch{MM'}} \otimes (\cdot)))}\right]\\
    &= \Trr[MM']{\phi^{+\ch{MM'}}(\Lambda_A(\phi^{+\ch{MM'}} \otimes (\cdot)))},
\end{align*}
where we have used the ``mirror lemma," $A^{\ch M}\ket\phi^{+\ch{MM'}}=A^{T\ch M'}\ket\phi^{+\ch{MM'}}$ in the first and third equality, and  $(\cdot)^T$ is the transpose with respect to the computational basis and $\bar{(\cdot)}$ is the complex conjugate. In the same way one can deduce that\\ $\Lambda_2 = \Trr[MM']{(\I - \phi^{+\ch{MM'}})(\Lambda_A(\phi^{+\ch{MM'}} \otimes (\cdot)))}$ and thus $\Pi$ is \NM.
\end{proof}

We can use this equivalence to adopt results proven for \NM in \cite{NonMal}, particularly that the unitaries in a unitary encryption scheme form a unitary 2-design.

\begin{definition}
    A family of unitary matrices $D$ is an \emph{$\eps$-approximate 2-design} if
    \[ \dn{\frac{1}{|D|}\sum\limits_{U\in D} (U\otimes U)(\cdot)(U^\dagger\otimes U^\dagger) - \int (U\otimes U)(\cdot)(U^\dagger\otimes U^\dagger) \mathop{dU}} \leq \eps.\]
\end{definition}

\begin{corollary}\label{thm:DNSfromPNM}
    Let $\Pi = (\KeyGen, \Enc, \Dec)$ be a unitary \SKQES such that $\Enc_k(\rho) = V_k\rho V_k^\dagger$ for some family of unitaries $D = \{V_k\}_k$ and $|M|=|C|=2^n$, then $\Pi$ being \PNM or \CiNM is equivalent to $D$ to being an approximate 2-design, in the sense that, for a sufficiently large constant $r$\footnote{For the exact value of $r$ and the constants hidden by the $\Omega$-s we refer to Theorem C.3 in \cite{NonMal} and Lemma 2.2.14 in \cite{low2010pseudo}},
    \begin{enumerate}
        \item If $D$ is a $\Omega(2^{-rn})$-approximate 2-design then $\Pi$ is $2^{-\Omega(n)}$-\PNM and $2^{-\Omega(n)}$-\CiNM.
        \item If $\Pi$ is $\Omega(2^{-rn})$-\PNM or $\Omega(2^{-rn})$-\CiNM, then $D$ is a $2^{-\Omega(n)}$-approximate 2-design.
    \end{enumerate}
\end{corollary}

To provide additional evidence that \PNM captures plaintext non-malleability for \SKQES in a satisfactory way, we show that any \PNM-secure scheme can be used to construct a plaintext-authenticating scheme in the sense of \cite{dupuis2012actively}, see Definition~\ref{def:dns}.
The intuition behind \DNS-authentication is that, after a possible attack, one can determine from a received plaintext whether or not an attack was performed, unless the attack did not change the underlying plaintext. For this reason, \DNS-authentication is a notion of plaintext authentication. We use the fact that a \PNM\ scheme protects a plaintext from modification to protect a tag register, which we then use to detect whether an attack was attempted. With this in mind, we first determine what state makes a good tag.

\begin{lemma}
\label{lem:dns_tag}
For any \SKQES\ $(\KeyGen, \Enc, \Dec)$ and any $m \in \mathbb{N}$ such that $M = M'R$ for some registers $M'$ and $R$ with $\log|R| = m$ there exists an $x \in \{0,1\}^m$ such that $\Tr[\bra{x}^{\ch{R}}\Dec_K(\tau^{\ch{C}})\ket{x}^{\ch{R}}] \leq \frac{1}{|R|}$.
\end{lemma}
\begin{proof}
Observe that
\begin{align*}
    \E\limits_{x\in\{0,1\}^m}[\Trr{\bra{x}\Dec_K(\tau^{\ch{C}})\ket{x}}] &= \sum\limits_{x\in\{0,1\}^m}\frac{1}{2^m} \Trr{\bra{x}\Dec_K(\tau^{\ch{C}})\ket{x}}\\
    &= \frac{1}{2^m} \Trr{\Dec_K(\tau^{\ch{C}})} = \frac{1}{|R|}\\
\end{align*}

Since the expected value of $\Trr{\bra{x}\Dec_K(\tau^{\ch{C}})\ket{x}}$ is $\frac{1}{|R|}$, there must be at least one $x$ such that $\Trr{\bra{x}\Dec_K(\tau^{\ch{C}})\ket{x}} \leq \frac{1}{|R|}$. 
\end{proof}

Lemma \ref{lem:dns_tag} allows us to find tags that have little overlap with $\Dec_K(\tau^{\ch{C}})$, which means one can distinguish well between the case were the tag was left unharmed and the case where the ciphertext was depolarized. We use this property to build a scheme that is \DNS\ authenticating.

\begin{theorem}
\label{thm:pnm_dns}
For any $\eps$-\PNM \SKQES $\Pi = (\KeyGen, \Enc, \Dec)$, there exists some $x$ such that the scheme $\Pi' = (\KeyGen, \Enc', \Dec')$ is $\left(\frac{3}{|R|} + \eps\right)$-\DNS-authenticating, where
\begin{align*}
    \Enc'_k &= \Enc_k((\cdot)^{\ch{M'}} \otimes \kb{x}{x}^{\ch{R}})\\
    \Dec'_k &= \bra{x}^{\ch{R}}\Dec_k(\cdot)\ket{x}^{\ch{R}} + \Trr{(\I^{\ch{R}} - \kb{x}{x}^{\ch{R}})\Dec_k(\cdot)}\kb{\bot}{\bot}
\end{align*}
\end{theorem}
\begin{proof}
We sketch the proof here, the full proof of this theorem can be found in Appendix \ref{thm:prf_pnm_dns}. Take $x$ as in Lemma \ref{lem:dns_tag}. Let $\Lambda_A$ be an arbitrary attack map on $\Pi'$, then its effective map is
\[ \Tilde{\Lambda}'_A = \Dec_{check} \circ \Tilde{\Lambda}_A \circ \Enc_{append}, \]
where $\Tilde{\Lambda}_A$ is the effective map of $\Lambda_A$ as an attack on $\Pi$ and $\Enc_{append}$ and $\Dec_{check}$ are the channels that perform adding $\kb{x}{x}$ to the plaintext during encryption and removing and checking of $\kb{x}{x}$ during decryption respectively. Since $\Pi$ is $\eps$-\PNM, there exist $\Lambda_1, \Lambda_2$ such that
\[ \dn{\Tilde{\Lambda}_A - \id \otimes \Lambda_1 + \frac{1}{|M|^2-1}(|M|^2\lr{\Dec_K(\tau^{\ch{C}})} - \id)^{\ch{M}}\otimes \Lambda_2}\leq \eps. \]
Define $\Lambda_{acc} = \Lambda_1$, $\Lambda_{rej} = \Lambda_2$, then
\[ \dn{\Tilde{\Lambda}'_A - \id \otimes \Lambda_{acc} - \lr{\kb{\bot}{\bot}} \otimes \Lambda_{rej}} \leq \eps + \frac{3}{|R|},\]
which means that $\Pi'$ is $\left(\frac{3}{|R|} + \eps\right)$-\DNS\ authenticating.
\end{proof}

Note that $|R|$ is a parameter of the scheme and any \PNM scheme (with negligible $\eps$) can be made into a \DNS scheme (with negligible $\eps$) by taking $|R| = 2^n$, i.e. taking $R$ as $n$ qubits.

\section{Non-Malleability for Quantum PKE} \label{sec:PKE}

\subsection{Quantum Comparison-Based Non-Malleability}

In this section, we will define a notion of many-time non-malleability for quantum public-key encryption, quantum comparison-based non-malleability (\QCNM), as a quantum analog of the classical notion of comparison-based non-malleability (\CNM, see Section \ref{sec:prelims}) introduced in \cite{bellare1999non}. We first analyze \CNM with the goal of finding appropriate quantum analogs of each of its components. 


The message distribution $M$ in the \CNM\ definition allows an adversary to select messages that she thinks might produce ciphertexts that can be modified in a structural way. This choice is given because the total plaintext space is exponentially large, thus if one picks a message completely at random and only a few of them can be modified into related ciphertexts, then the winning probability is negligible despite the scheme being insecure. In the quantum representation of this message space we consider the following requirements:
\begin{enumerate}
    \item As mentioned earlier, the quantum no-cloning theorem prevents copying the plaintext after sampling it for future reference. In order to check the relation in the last step of \CNM, we require that two related states are produced, one of which will be kept by the challenger and the other encrypted and used by the adversary.
	\item It must not be possible for the adversary to correlate herself with either of the produced messages. This is to prevent the adversary from influencing the second copy of the state later on. For example, consider the case where the adversary produces the state $\frac{1}{\sqrt{2}}(\ket{000} + \ket{111})$, where the first two qubits are the two copies of the message and the last is kept by the adversary. The adversary can then measure her qubit, collapsing the superposition and informing her in which (classical) state the second copy now is. This allows her to trivially construct a relation between her output and the second copy.
\end{enumerate}

In order to satisfy requirement (1), we have chosen to represent $M$ by a unitary $U^{\ch{MRP}}$ such that $U\ket{0}$ is a purification of the message distribution, where the message resides in $M$, the second (reference) state in $R$, and $P$ is used for the purification\footnote{A purification is a quantum register that is similar to the ``garbage'' register in reversible computation.}. This purification register allows the adversary to implement any quantum channel on $MR$, with $U$ being a Stinespring dilation of that channel. The first part of the quantum adversary, $\MA_1$, produces this unitary in the form of a circuit along with some side information $S$ to be passed on to the next stage. We denote this process by $(U, S) \leftarrow \MA_1(pk)$. 

For the \QCNM\ definition we define two experiments, similar to the \CNM\ definition. In the following we describe  how the different elements of the \CNM experiment are instantiated in the quantum case. The appropriate quantum notion of a relation $R$ on plaintexts is given by a POVM element  $0\le E^{\ch{MR}}\le \I$. The two registers $MR$ are considered to contain related states if an application of the measurement $\{E, \I-E\}$ returns the outcome corresponding to $E$. Of course, this POVM is provided by the adversary in form of a circuit and must hence be efficient. The quantum analogue of the vector $\mathbf{y}$ is given by a collection of registers $\textbf{C} = C_1\dots C_m$, where $m$ is at most polynomial in $n$, the security parameter of the considered scheme, and each $C_i$ satisfies $M_iT_i = C_i \cong C = MT$. The quantum analogue of the vector $\mathbf{x}$ is similarly given as $\mathbf{M} = M_1\dots M_m$. Observe that any $\PKQES$ can also be seen as a $\SKQES$, with keys of the form $k = (pk, sk)$, which allows us to use Theorem \ref{thm:charac}\footnote{This characterization could also be invoked with $k = pk$, however the resulting encryption unitary is then (likely) not efficiently implementable.}. For any \PKQES\ $\Pi = (\KeyGen, \Enc, \Dec)$ with security parameter $n$, let $\{V_k \mid k = (pk,sk) \leftarrow \KeyGen(1^n)\}$, $t = \log|T|$, $\{\psi_{k,r} \mid k = (pk,sk) \leftarrow \KeyGen(1^n), r \in \{0,1\}^t\}$ and $\{p_k \mid k = (pk,sk) \leftarrow \KeyGen(1^n)\}$ be as in Theorem \ref{thm:charac} in the \QCNM experiments.

Lastly, we define the unitary $U_{prep}$ combining the preparation of the message state and the encryption of its part in register $M$. This means a check similar to the $y\not\in \mathbf{y}$ check in the \CNM\ experiments can be implemented by sequentially undoing $U_{prep}$ on all combinations $C_iRP$ and then measuring whether the result is $\kb{0}{0}$, which is only the case if $C_i$ contained part of $U_{prep}\ket{0}$, which is the original ciphertext given to the adversary.

We are now ready to define the real and ideal experiments for quantum comparison-based non-malleability.

\begin{algorithm}[H]
	\caption{\QCNMreal}\label{exp:QCNMr}
\DontPrintSemicolon
\SetKwInOut{Input}{Input}\SetKwInOut{Output}{Output}
\Input{$\Pi, \MA, n$}
\Output{$b \in \{0,1\}$}
\BlankLine
$k = (pk, sk) \leftarrow \KeyGen(1^n)$\;
$(U^{\ch{MRP}}, S) \leftarrow \MA_1(pk)$\;
$r \xleftarrow{p_k} \{0,1\}^t$\;
Construct $U_\psi^{\ch{T}}$ such that $U_\psi^{\ch{T}}\ket{0}^{\ch{T}} = \ket{\psi_{k,r}}^{\ch{T}}$\;
Construct $U_{prep}^{\ch{MTRP}} = V_{k}^{\ch{MT}}(U^{\ch{MRP}} \otimes U_\psi^{\ch{T}})$\;
Prepare $U_{prep}\kb{0}{0}U_{prep}^\dagger$ in $MTRP$\;
$(\mathbf{C}, E) \leftarrow \MA_2(MT, S)$\;
\For{$i = 1,\dots, |\mathbf{C}|$}{
    Perform $U_{prep}^\dagger$ on $C_iRP$\;
    Measure $\{\kb{0}{0}, \I - \kb{0}{0}\}$ on $C_iRP$ with outcome $b$\;
    \If{$b = 0$}{Output $0$}
    Perform $U_{prep}$ on $C_iRP$\;
}
$\mathbf{M} \leftarrow \Dec_{sk}(\mathbf{C})$\;
$\{E, \I - E\}$ on $R\mathbf{M}$ with outcome $e$\;
Output $e$
\end{algorithm}
\begin{algorithm}[H]
	\caption{\QCNMideal}\label{exp:QCNMi}
\DontPrintSemicolon
\SetKwInOut{Input}{Input}\SetKwInOut{Output}{Output}
\Input{$\Pi, \MA, n$}
\Output{$b \in \{0,1\}$}
\BlankLine
Run lines 1-14 
of Experiment \QCNMreal\;
\setcounter{AlgoLine}{14}
Prepare $U\kb{0}{0}U^\dagger$ in $\Tilde{M}\Tilde{R}\Tilde{P}$
\;
$\{E, \I - E\}$ on $
\Tilde{R}
\mathbf{M}$ with outcome $e$\;
Output $e$
\end{algorithm}

A \PKQES is now defined to be \QCNM-secure, if no adversary can achieve higher success probability in the experiment \QCNMreal than in \QCNMideal.
\begin{definition}
A $\PKQES$ $\Pi$ is \emph{quantum comparison-based non-malleable \\(\QCNM)} if for any QPT adversary $\MA = (\MA_1, \MA_2)$ it holds that
\[ \Prr{\QCNMreal(\Pi, \MA, n) = 1} - \Prr{\QCNMideal(\Pi, \MA, n) = 1} \leq \negl(n),\]
if $\MA$ such that:
\begin{itemize}
    \item $\MA_1$ outputs a valid unitary $U$ which can be implemented by a QPT algorithm, 
    \item $\MA_2$ outputs a POVM element $E$ which can be implemented by a QPT algorithm, 
    \item $\MA_2$ outputs a vector of registers $\mathbf{C}$ such that $\bot \not\in \Dec_{sk}(\mathbf{C})$.
\end{itemize}
\end{definition}


\subsection{Relation Between \QCNM\ and \CNM}
\label{sec:qcnm_cnm}

In order to compare \QCNM to \CNM, we consider both definitions modified for quantum adversaries and encryption schemes that have classical input and output but can perform quantum computation. In the case that a quantum state is sent to such a post-quantum algorithm, it is first measured in the computational basis to obtain a classical input. 

\begin{algorithm}[H]
	\caption{$\QCNMreal_{PQ}$}\label{exp:pqQCNMr}
\DontPrintSemicolon
\SetKwInOut{Input}{Input}\SetKwInOut{Output}{Output}
\Input{$\Pi, \MA, n$}
\Output{$b \in \{0,1\}$}
\BlankLine
$k = (pk, sk) \leftarrow \KeyGen(1^n)$\;
$(U^{\ch{MRP}}, S) \leftarrow \MA_1(pk)$\;
$r \xleftarrow{p_k} \{0,1\}^t$\;
Construct $U_\psi^{\ch{T}}$ such that $U_\psi^{\ch{T}}\ket{0}^{\ch{T}} = \ket{\psi_{k,r}}^{\ch{T}}$\;
Construct $U_{prep}^{\ch{MTRP}} = V_{k}^{\ch{MT}}(U^{\ch{MRP}} \otimes U_\psi^{\ch{T}})$\;
Prepare $U_{prep}\kb{0}{0}U_{prep}^\dagger$ in $MTRP$\;
Measure $MTR$ in the computational basis with outcome $y^{\ch{MT}}z^{\ch{R}}$\;
$(\mathbf{C}, E) \leftarrow \MA_2(MT, S)$\;
\For{$i = 1,\dots, |\mathbf{C}|$}{
    Measure $\{\kb{y}{y}, \I - \kb{y}{y}\}$ on $C_i$ with outcome $b$\;
    \If{$b = y$}{Output $0$}
}
$\mathbf{M} \leftarrow \Dec_{sk}(\mathbf{C})$\;
$\{E, \I - E\}$ on $R\mathbf{M}$ with outcome $e$\;
Output $e$
\end{algorithm}

\begin{algorithm}[H]
	\caption{$\QCNMideal_{PQ}$}\label{exp:pqQCNMi}
\DontPrintSemicolon
\SetKwInOut{Input}{Input}\SetKwInOut{Output}{Output}
\Input{$\Pi, \MA, n$}
\Output{$b \in \{0,1\}$}
\BlankLine
Run lines 1-13 of Experiment $\QCNMideal_{PQ}$\;
\setcounter{AlgoLine}{13}
Prepare $U\kb{0}{0}U^\dagger$ in $\Tilde{M}\Tilde{R}\Tilde{P}$\;
Measure $\Tilde{M}\Tilde{R}$ in the computational basis\;
$\{E, \I - E\}$ on $\Tilde{R}\mathbf{M}$ with outcome $e$\;
Output $e$
\end{algorithm}

We consider the above experiments to be the post-quantum version of the \QCNM\ experiments. The main modification is the measurement in Step 7, which enforces the requirement that $\MA_2$ only takes classical input. The modification of Steps 9 through 12 is made because the measurement in Step 7 disturbs the state in an irreversible fashion, thus performing $U_{prep}^\dagger$ no longer inverts the sampling/encryption process. Lastly, in the Ideal setting Step 15 is added to mimic the effect that Step 7 would have on $M'$.


\begin{definition}
A $\PKQES$ $\Pi$ is \emph{post-quantum comparison-based non-malleable ($\QCNM_{PQ}$)} if for any adversary $\MA = (\MA_1, \MA_2)$ it holds that
\[ \Prr{\QCNMreal_{PQ}(\Pi, \MA, n) = 1} - \Prr{\QCNMideal_{PQ}(\Pi, \MA, n) = 1} \leq \negl(n),\]
if $\MA$ and $\Pi$ are such that:
\begin{itemize}
    \item $\MA_1$ and $\MA_2$ are QPT and output only classical states,
    \item $\MA_1$ outputs a valid unitary $U$ which can be implemented by a QPT algorithm, 
    \item $\MA_2$ outputs a POVM element $E$ which implementable by a QPT algorithm, and
    \item $\MA_2$ outputs a vector of registers $\mathbf{C}$ such that $\bot \not\in \Dec_{sk}(\mathbf{C})$.
\end{itemize}
\end{definition}

Sampling of the message by the challenger is now done by not only applying $U$ to $\ket{0}$, but in addition also measuring in the computational basis. 
Similarly, we define a post-quantum version of \CNM.
\begin{definition}
A $\PKQES$ $\Pi$ is \emph{comparison-based non-malleable against post-quantum adversaries ($\CNM_{PQ}$)} if for any \QPT adversary $\MA = (\MA_1, \MA_2)$ it holds that
\[ \Prr{\CNMreal(\Pi, \MA, n) = 1} - \Prr{\CNMideal(\Pi, \MA, n) = 1} \leq \negl(n),\]
if $\Pi$ and $\MA$ are such that: 
\begin{itemize}
    \item $\MA_1$ and $\MA_2$ are \QPT and output classical strings,
    \item $\MA_1$ outputs a valid \QPT algorithm $M$ which produces classical strings,
    \item $\MA_2$ outputs a \QPT algorithm $R$,
    \item $\MA_2$ outputs a vector $\mathbf{y}$ such that $\bot \not\in \Dec_{sk}(\mathbf{y})$.
\end{itemize}
\end{definition}

The only difference between \CNM\ and $\CNM_{PQ}$ is that the latter allows
the encryption scheme, adversary and any algorithms produced by the adversary to use a quantum computer.  Furthermore, the relation $R$ has become probabilistic, but since it is used only once there is no difference between using a probabilistic relation or picking a deterministic relation at random. Observe that $\CNM_{PQ}$ is simply a stronger requirement than \CNM\ since it requires security against a strict superset of adversaries, and thus trivially implies \CNM.

\begin{theorem}\label{thm:QCNMequivPQ}
A \PKQES\ $\Pi$ is $\QCNM_{PQ}$ if and only if $\Pi$ is $\CNM_{PQ}$.
\end{theorem}
\begin{proof}
    For the $\Rightarrow$ direction, let $\Pi$ be an arbitrary \PKQES\ which is $\QCNM_{PQ}$-secure and let $\MA = (\MA_1, \MA_2)$ be an arbitrary quantum adversary intended to perform the $\CNM_{PQ}$ experiments. Assume that $\Pi$ is such that $\Enc$ and $\Dec$ take only classical input and produce only classical output. Define $\MB = (\MB_1, \MB_2)$ as follows:\\
    {
    \setlength{\interspacetitleruled}{0pt}
    \setlength{\algotitleheightrule}{0pt}
    $\MB_1(pk)$:\\
    \begin{algorithm}[H]
    \DontPrintSemicolon
    \SetKwInOut{Input}{Input}\SetKwInOut{Output}{Output}
    $(M,s) \leftarrow \MA_1(pk)$\;
    Let $p_M(x)$ be the probability that $x \leftarrow M$, then construct $U$ such that \( U\ket{0}^{\ch{MM'P}} = \frac{1}{|R|}\sum\limits_{r \in R}\ket{M(r)M(r)r} = \sum\limits_{x \leftarrow M} \sqrt{p_M(x)}\ket{xx\phi_x}^{\ch{MM'P}}, \)
    where $M' \cong M$ is the reference register, $R$ is the set of possible input for $M$ and $\phi_x$ is the uniform superposition over all $\ket{r}$ such that $x\leftarrow M(r)$.\;
    Output $(U, \kb{s}{s})$
    \end{algorithm}
    }
    
    {
    \setlength{\interspacetitleruled}{0pt}
    \setlength{\algotitleheightrule}{0pt}
    $\MB_2(\kb{s}{s}^{\ch{S}}, \kb{y}{y}^{\ch{MT}})$:\\
    \begin{algorithm}[H]
    \DontPrintSemicolon
    \SetKwInOut{Input}{Input}\SetKwInOut{Output}{Output}
    $(R, \mathbf{y}) \leftarrow \MA_2(y, s)$\;
    Construct $E = \sum\limits_{i,\mathbf{j}}R(i,\mathbf{j})\kb{i\mathbf{j}}{i\mathbf{j}}$\;
    Output $(E, \kb{\mathbf{y}}{\mathbf{y}}^{\ch{C_1\dots C_m}})$
    \end{algorithm}
    }
    
    Observe that the definition of $\QCNMreal_{PQ}(\Pi, \MB, n)$, after some simplification, yields\\
    {
    \setlength{\interspacetitleruled}{0pt}
    \setlength{\algotitleheightrule}{0pt}
    \begin{algorithm}[H]
    \DontPrintSemicolon
    \SetKwInOut{Input}{Input}\SetKwInOut{Output}{Output}
    $k = (pk, sk) \leftarrow \KeyGen(1^n)$\;
    $(M,s) \leftarrow \MA_1(pk)$\;
    Let $p_M(x)$ be the probability that $x \leftarrow M$, then construct $U$ such that $U\ket{0}^{\ch{MM'P}} = \sum\limits_{x \leftarrow M} \sqrt{p_M(x)}\ket{xx\phi_x}^{\ch{MM'P}}$\;
    $r \xleftarrow{p_k} \{0,1\}^t$\;
    Construct $U_\psi^{\ch{T}}$ such that $U_\psi^{\ch{T}}\ket{0}^{\ch{T}} = \ket{\psi_{k,r}}^{\ch{T}}$\;
    Construct $U_{prep}^{\ch{MTM'P}} = V_{k}^{\ch{MT}}(U^{\ch{MM'P}} \otimes U_\psi^{\ch{T}})$\;
    Prepare $U_{prep}\kb{0}{0}U_{prep}^\dagger$ in $MTM'P$\;
    Measure $MTM'$ in the computational basis with outcome $yz$\;
    $(R, \mathbf{y}) \leftarrow \MA_2(y, s)$\;
    Construct $E = \sum\limits_{i,\mathbf{j}}R(i,\mathbf{j})\kb{i\mathbf{j}}{i\mathbf{j}}$\;
    Prepare $\kb{\mathbf{y}}{\mathbf{y}}$ in $\mathbf{C}$\;
    \For{$i = 1,\dots, |\mathbf{C}|$}{
        Measure $\{\kb{y}{y}, \I - \kb{y}{y}\}$ on $C_i$ with outcome $b$\;
        \If{$b = y$}{Output $0$}
    }
    $\mathbf{M} \leftarrow \Dec_{sk}(\mathbf{C})$\;
    $\{E, \I - E\}$ on $M'\mathbf{M}$ with outcome $e$\;
    Output $e$
    \end{algorithm}
    }
    Here Steps 3,5,6,7 and 8 together simply execute $x\leftarrow M;y\leftarrow\Enc_{k;r}(x)$. Furthermore, if $y\in\mathbf{y}$ then some $C_i$ contains $\kb{y}{y}$, which will guarantee the output to be $y$ in Step 13. Conversely if $y\not\in\mathbf{y}$, then all $C_i$ contain some state orthogonal to $\kb{y}{y}$ and thus Step 13 has $0$ probability of outputting $y$ in this case, thus Step 13 effectively implements the $y\not\in\mathbf{y}$ check. Lastly, note that $E$ is a projective measurement which projects onto the space spanned by all $\ket{i\mathbf{j}}$ such that $R(i,\mathbf{j})$, which means that Step 17 outputs $1$ iff $R(x, \mathbf{x})$, where $x$ is stored in $M'$ and $\mathbf{x}$ in $\mathbf{M}$. We conclude that $\QCNMreal_{PQ}(\Pi, \MB, n)$ produces the same random variable as $\CNMreal(\Pi, \MA, n)$. By similar reasoning the same is true for the $\axname{Ideal}$ case, with the additional observation that preparing $U\ket{0}$ in $\Tilde{M}\Tilde{M}'\Tilde{P}$ and measuring $\Tilde{M}$ in the computational basis with result $\Tilde{x}$ is equivalent to $\Tilde{x} \leftarrow M$ and collapses $\Tilde{M}'$ to $\Tilde{x}$. 
    It follows that $\Pi$ is $\CNM_{PQ}$.
    
    For the $\Leftarrow$ direction, let $\Pi$ be an arbitrary PKQES\ fulfilling $\CNM_{PQ}$ and let $\MA = (\MA_1, \MA_2)$ be an arbitrary classical adversary on this scheme intended to perform the $\QCNM_{PQ}$ experiments. Define $\MB = (\MB_1, \MB_2)$ as follows:\\
    {
    \setlength{\interspacetitleruled}{0pt}
    \setlength{\algotitleheightrule}{0pt}
    $\MB_1(pk)$:\\
    \begin{algorithm}[H]
    \DontPrintSemicolon
    \SetKwInOut{Input}{Input}\SetKwInOut{Output}{Output}
    \SetKwBlock{Sub}{}{}
    $(U, \kb{s}{s}^{\ch{S}}) \leftarrow \MA_1(pk)$\;
    Prepare $U\ket{0}$ twice, in $M_0R_0P_0$ and $M_1R_1P_1$\;
	Measure $M_0R_0M_1R_1$ in the computational basis with outcome $m_0z_0m_1z_1$\;
	Construct $M$ to be the uniform distribution over $\{m_0, m_1\}$\;
    Output $(M,sm_0z_0m_1z_1)$
    \end{algorithm}
    }
    {
    \setlength{\interspacetitleruled}{0pt}
    \setlength{\algotitleheightrule}{0pt}
    \noindent$\MB_2(s', y)$:\\
    \begin{algorithm}[H]
    \DontPrintSemicolon
    \SetKwInOut{Input}{Input}\SetKwInOut{Output}{Output}
    \SetKwBlock{Sub}{}{}
	Parse $s'$ as $sm_0z_0m_1z_1$\;
    $(E, \mathbf{y}) \leftarrow \MA_2(\kb{s}{s}^{\ch{S}}, \kb{y}{y}^{\ch{MT}})$\;
    Construct $R(x, \mathbf{x})$ to be
    \Sub{
		Find $i$ such that $m_i=x$\;
        \textbf{prepare} $\kb{z_i\mathbf{x}}{z_i\mathbf{x}}$ in $R'\mathbf{M}$\;
        \textbf{measure} $\{E, \I - E\}$ on $R'\mathbf{M}$, \textbf{output} $1$ iff the outcome is $E$
    }
    Output $(R, \kb{\mathbf{y}}{\mathbf{y}})$
    \end{algorithm}
    }
    
    Observe that the definition of $\CNMreal_{PQ}(\Pi, \MB, n)$, after some simplification, yields\\
    {
    \setlength{\interspacetitleruled}{0pt}
    \setlength{\algotitleheightrule}{0pt}
    \begin{algorithm}[H]
    \DontPrintSemicolon
    \SetKwInOut{Input}{Input}\SetKwInOut{Output}{Output}
    $k = (pk, sk) \leftarrow \KeyGen(1^n)$\;
    $(U, \kb{s}{s}^{\ch{S}}) \leftarrow \MA_1(pk)$\;
    Prepare $U\ket{0}$ in $M_0R_0P_0$\;
    Prepare $U\ket{0}$ in $M_1R_1P_1$\;
    Measure $M_0R_0M_1R_1$ in the computational basis with outcome $m_0z_0m_1z_1$\;
    Pick $i \leftarrow \{0,1\}$\;
    $y \leftarrow \Enc_{pk}(m_i)$\;
    $(E, \kb{\mathbf{y}}{\mathbf{y}}) \leftarrow \MA_2(\kb{s}{s}^{\ch{S}}, \kb{y}{y}^{\ch{MT}})$\;
    $\mathbf{x} \leftarrow \Dec_{sk}(\mathbf{y})$\;
    \If{$y\in\mathbf{y}$}{Output $0$}
    Find $j$ such that $m_j = m_i$\;
    Prepare $\kb{z_j\mathbf{x}}{z_j\mathbf{x}}$ in $R\mathbf{M}$\;
    $\{E, \I - E\}$ on $R\mathbf{M}$ with outcome $e$\;
    Output $e$
    \end{algorithm}
    }
    
    Similarly, the $\CNMideal_{PQ}(\Pi, \MB, n)$ yields the same Experiment except with line 12 replaced with ``Pick $j \leftarrow \{0,1\}$''. Note that Step 7, the encrypting, is not performed by $U_{prep}$ but simply by $\Enc$ and that Step 10 simply checks $y\in\mathbf{y}$ instead of loop that we earlier argued to be equivalent. Additionally the measurement in Step 3 and 4 are equivalent to measuring the ciphertext after encryption (as is done in \QCNM), because it is assumed that encryption, and thus $V_k$, maps classical states to classical states.

    Note that w.l.o.g. we can assume that $m_0 \neq m_1$, since if this is not the case then the Real and Ideal case are equivalent and thus the adversary has no hope of winning. This makes that the $\CNMreal_{PQ}(\Pi, \MB, n)$ and $\QCNMreal_{PQ}(\Pi, \MA, n)$ are equivalent given the observations in the previous paragraph. Furthermore, when $i=j$ in the $\CNMideal$ case then it is equivalent to the $\CNMreal$ case. When $i\neq j$, the $\CNMideal_{PQ}(\Pi, \MB, n)$ and $\QCNMideal_{PQ}(\Pi, \MA, n)$ experiments are equivalent. Thus the advantage of $\MB$ in $\CNM$ is half the advantage of $\MA$ in $\QCNM$, which implies that $\Pi$ is $\QCNM_{PQ}$.
\end{proof}

Note that we argued earlier that, for any \PKES, being $\CNM_{PQ}$ trivially implies being $\CNM$, thus we derive the following corollary.

\begin{corollary}
Any $\QCNM_{PQ}$ \PKES\ is \CNM.
\end{corollary}

\subsection{A \QCNM Secure Scheme}

In this section we show how \QCNM-security can be achieved using a quantum-classical hybrid construction like the ones used in \cite{alagic2018unforgeable,SignCrypt}.
The idea is similar to the classical technique of hybrid encryption. We construct a quantum-non-malleable \PKQES by encrypting each plaintext with a quantum one-time non-malleable scheme and encrypting the key using a classical non-malleable \PKES. We begin by defining the general quantum-classical hybrid construction.

\begin{construction}
	Let $\Pi^\Qu=(\KeyGen^\Qu,\Enc^\Qu, \Dec^\Qu)$ be a \SKQES and $\Pi^\Cl=(\KeyGen^\Cl,\Enc^\Cl, \Dec^\Cl)$ a \PKES. We define the hybrid scheme $\Pi^\Hyb[\Pi^\Qu,\Pi^\Cl]=(\KeyGen^\Hyb,\Enc^\Hyb \Dec^\Hyb)$ as follows. We set $\KeyGen^\Hyb=\KeyGen^\Cl$. The encryption algorithm $\Enc^\Hyb_\pk$, on input $X$,
	\begin{enumerate}
		\item generates a key $k\leftarrow\KeyGen^\Qu(1^{n(\pk)})$, and
		\item outputs the pair $(\Enc^\Qu_k(X), \Enc^\Cl_\pk(k))$.
	\end{enumerate}
Decryption is done in the obvious way, by first decrypting the second part of the ciphertext using $\Dec^\Cl$ to obtain the one-time key $k'$, and then decrypting the first part using $\Dec^\Qu_{k'}$.
\end{construction}

We continue by proving that if $\Pi^\Qu$ is unitary and secure according to \NM, \CiNM or \PNM (they are all equivalent for unitary \SKQES according to Theorem  \ref{thm:unitaryissimple}), and $\Pi^\Cl$ to be \CNM, then $\Pi^\Hyb[\Pi^\Qu,\Pi^\Cl]$ is \QCNM.
\begin{theorem}\label{thm:security}
	Let $\Pi^\Qu=(\KeyGen^\Qu,\Enc^\Qu, \Dec^\Qu)$ be a \NM secure  \SKQES with unitary encryption and decryption map, and $\Pi^\Cl=(\KeyGen^\Cl,\Enc^\Cl, \Dec^\Cl)$ a postquantum-\CNM secure \PKES. Then $\Pi^\Hyb[\Pi^\Qu,\Pi^\Cl]$ is \QCNM.
\end{theorem}
\begin{proof} 
	We begin by defining modified versions of the two experiments used in defining \QCNM, \sckQCNMreal and \sckQCNMideal (for \textbf spoofed \textbf classical \textbf key). These two experiments are defined exactly as the experiments \QCNMreal and \QCNMideal, except for the following modifications:
	\begin{enumerate}
		\item When creating the ciphertext register $C$ that is handed to the adversary, its classical part $c$ is produced by encrypting a fresh, independently sampled one-time key $k'\from\KeyGen^\Qu$. The pair $(c, k)$ is stored ($k$ being the key used for encryption with $\Enc^\Qu$.)
		\item The test whether the ciphertext was modified by the adversary is done by first checking whether the classical part $c'$ is equal to $c$. 
		If it is not, the ciphertext was modified and no further test of the quantum part is necessary. If $c'=c$, the modification check from the games \QCNMreal and \QCNMideal is applied, using the stored one-time key $k$. Note that this is equivalent to the check mandated for the \QCNM experiments.
		\item Before decrypting any ciphertext, the challenger checks whether its classical part is equal to $c$. If not, he proceeds with decryption, otherwise, he just decrypts the quantum ciphertext with $\Dec^\Qu_k$.
	\end{enumerate}
	
	Let \adver be a \QCNM-adversary against $\Pi^\Hyb$. Recall that it was proven in \cite{bellare1999non} that \CNM is equivalent to \INDparCCAA, indistinguishability under parallel chosen ciphertext attacks. In this attack model, after receiving the challenge ciphertext, the adversary is allowed to submit one tuple of ciphertexts that is decrypted in case none of them is equal to the challenge ciphertext. Define the following \INDparCCAA adversary $\adver'$ against $\Pi^\Cl$. $\adver'$ simulates the $\QCNMreal(\Pi^\Hyb,\adver,n)$-experiment. When the \QCNMreal challenger is supposed to encrypt a plaintext to be sent to \adver, $\adver'$ sends $m_0=k$ and $m_1=k'$ as challenge plaintexts to the \INDparCCAA challenger, where $k,k'\from\KeyGen^{\Qu}$, and $k$ is used to encrypt the quantum plaintext. After storing a copy of the resulting classical ciphertext $c$ and the one-time key $k$, $\adver'$ continues to simulate $\QCNMreal(\Pi^\Hyb,\adver,n)$ but using the mixed quantum-classical modification check from the spoofed classical key experiments defined above. Decryption is done using the \parCCAA oracle, except for the ciphertexts with classical part $c$, which are just decrypted using the stored one-time key $k$. Now $\adver'$ outputs the result of the simulated experiment $\QCNMreal(\Pi^\Hyb,\adver,n)$.
	
	Now observe that if the \INDparCCAA challenger's bit comes up $b=0$, $\adver'$ faithfully simulated the experiment $\QCNMreal(\Pi^\Hyb,\adver,n)$, while the case $b=1$ results in a simulation of $\sckQCNMreal(\Pi^\Hyb,\adver,n)$. Therefore, the \INDparCCAA security of $\Pi^\Cl$ implies that the games \QCNMreal and \sckQCNMreal have the same result, up to negligible difference.
	
	We can also define an \INDparCCAA adversary $\adver''$ against $\Pi^\Cl$ in the same way as $\adver'$, but this time using the \QCNMideal experiments. This implies analogously that the experiments \QCNMideal and \sckQCNMideal also have the same result, up to negligible difference.
	
	Finally, what is left to prove is that the experiments \sckQCNMreal and \sckQCNMideal have the same outcome due to the \NM security of $\Pi^\Qu$. If the classical part of the ciphertext has been modified, $\Dec_k$ is never applied. By the fact that the scheme $\Pi^\Qu$ is \IND secure \cite{NonMal}, $\mathbf M$ is independent of (i.e. in a product state with) $R$, i.e. $\mathbf MR$ and $\mathbf M\tilde R$ have the same state. Therefore, \sckQCNMreal and \sckQCNMideal have the same outcome. 
	For the remaining case of $c'=c$, note that the modification test in lines 8 through 13 of Experiments \ref{exp:QCNMr} and \ref{exp:QCNMi} are identical, and that the application of $\left(V^{\ch{M}}\right)^\dagger$ ($T$ is trivial for unitary encryption) is equal to decryption. We can hence decrypt all ciphertexts before the modification test in the experiments \sckQCNMreal and \sckQCNMideal (line 9 in experiments \QCNMreal and \QCNMideal), and replace $U_{prep}$ by $U$. It follows that the rest of the experiment after decryption does not depend on the one-time key $k$ anymore. Hence the experiment has the form of a multi-decryption attack on the scheme $\Pi^\Qu$, i.e. where one ciphertext (the one that $\mathcal A_2$ receives as input) is mapped to many ciphertexts (the ones in $\mathbf C$) and are subsequently decrypted. We can therefore apply Lemma \ref{lem:CiNM-multidecrypt} to conclude that the modification test outputs $0$ unless $\mathbf M$ is in product with $R$, in which case $\mathbf MR$ and $\mathbf M\tilde R$ have the same state. \sckQCNMreal and \sckQCNMideal therefore have the same outcome. 
\end{proof}


\section{Open Questions}
After providing the first definition of non-malleability for quantum public-key encryption and showing how to fulfill it, and providing a comprehensive taxonomy of one-time security notions in the symmetric-key case, our work leaves a number of interesting open questions.

First, one might wonder what other connections \PNM and \CiNM have to other established security notions, such as the suggestion made in \cite{NonMal} that \NM, \CiNM or \PNM might be used to construct a totally authenticating scheme as defined in \cite{garg2017new}.

Second, many interesting problems remain in the computational setting. While our proposed definition of \QCNM provides a natural extension of \CNM to the quantum setting, a number of alternative but equivalent definitions of classical non-malleability exist, such as simulation-based non-malleability as defined in \cite{bellare1999non}. Besides the natural question whether \QCNM truly captures non-malleability, one might want to consider quantum versions of other classical notions of non-malleability and the relations between them. Furthermore, a symmetric-key version of \QCNM could be explored, which we suspect to be distinct from a computational version of \CiNM due to the mismatch of the way side information is handled (the $S$ register in \QCNM and the $B$ register in \CiNM).

\subsection*{Acknowledgements}
The authors thank the anonymous referees for the encouragement to attempt a better informal explanation of our results. CM would like to thank Gorjan Alagic for 
stimulating discussions. CM and CS were funded by a NWO VIDI grant (Project No. 639.022.519). CM was supported by a NWO VENI grant (Project No. VI.Veni.192.159).
JvW acknowledges the support from the Luxembourg National Research Fund via the CORE project Q-CoDe (Project No. 11689058).


\addcontentsline{toc}{section}{References}
\emergencystretch=1em
\printbibliography
\emergencystretch=0pt
\appendix
\section{Proofs}
\begin{lemma}[Lemma \ref{lem:CtoM}]
\label{lem:prf_CtoM}
Let $\Pi = (\KeyGen, \Enc, \Dec)$ be an arbitrary\\ \SKQES\ and $\Lambda_A^{\ch[CB]{C\hat{B}}}$ an arbitrary attack on $\Pi$ with effective map $\Tilde{\Lambda}_A^{\ch[MB]{M\hat{B}}}$. If there exist \CPTNI\ $\Lambda_1, \Lambda_2$, such that $\Lambda_1 + \Lambda_2$ is \CPTP and it holds that
\[ \dn{\Tilde{\Lambda}_A - \left(\id^{\ch{M}} \otimes \Lambda_1^{\ch[B]{\hat{B}}} + \frac{1}{|C|^2 - 1}\left(|C|^2\lr{\Dec_K(\tau^{\ch{C}})} - \id\right)^{\ch{M}} \otimes \Lambda_2^{\ch[B]{\hat{B}}}\right)} \leq \eps,\]
then for any $\alpha$ such that $|M|^2 \leq \alpha \leq |C|^2$ there exist \CPTNI\ $\Lambda_3, \Lambda_4$ such that $\Lambda_3 + \Lambda_4$ is \CPTP\ and
\[ \dn{\Tilde{\Lambda}_A - \left(\id^{\ch{M}} \otimes \Lambda_3^{\ch[B]{\hat{B}}} + \frac{1}{\alpha - 1}\left(\alpha\lr{\Dec_K(\tau^{\ch{C}})} - \id\right)^{\ch{M}} \otimes \Lambda_4^{\ch[B]{\hat{B}}}\right)} \leq \eps.\]
\end{lemma}
\begin{proof}
Assume that for some \CPTNI\ $\Lambda_1, \Lambda_2$ such that $\Lambda_1 + \Lambda_2$ is \CPTP\ it holds that
\[ \dn{\Tilde{\Lambda}_A - \left(\id^{\ch{M}} \otimes \Lambda_1^{\ch[B]{\hat{B}}} + \frac{1}{|C|^2 - 1}\left(|C|^2\lr{\Dec_K(\tau^{\ch{C}})} - \id\right)^{\ch{M}} \otimes \Lambda_2^{\ch[B]{\hat{B}}}\right)} \leq \eps.\]

Define $\gamma = \frac{(\alpha - 1)|C|^2}{\alpha(|C|^2-1)}$, $\Lambda_3 = \Lambda_1 + (1 - \gamma)\Lambda_2$, and $\Lambda_4 = \gamma\Lambda_2$. Note that $0 < \gamma \leq 1$ as long as $1 < \alpha \leq |C|^2$ and thus $\Lambda_3$ and $\Lambda_4$ are \CPTNI. Furthermore $\Lambda_3 + \Lambda_4 = \Lambda_1 + \Lambda_2$, thus $\Lambda_3 + \Lambda_4$ is \CPTP. Observe that

\begin{align*}
    \id^{\ch{M}} &\otimes \Lambda_3^{\ch[B]{\hat{B}}} + \frac{1}{\alpha - 1}\left(\alpha\lr{\Dec_K(\tau^{\ch{C}})} - \id\right)^{\ch{M}} \otimes \Lambda_4^{\ch[B]{\hat{B}}}\\
    &= \id^{\ch{M}} \otimes \left(\Lambda_1 + \left(1 - \gamma\right)\Lambda_2\right) + \frac{1}{\alpha-1}(\alpha\lr{\Dec_K(\tau^{\ch{C}})} - \id)^{\ch{M}} \otimes \gamma\Lambda_2\\
    &= \id^{\ch{M}} \otimes \Lambda_1 + (1-\gamma)\id^{\ch{M}} \otimes \Lambda_2 + \frac{\gamma}{\alpha - 1}(\alpha\lr{\Dec_K(\tau^{\ch{C}})} - \id)^{\ch{M}} \otimes \Lambda_2\\
    &= \id^{\ch{M}} \otimes \Lambda_1 + \frac{|C|^2}{|C|^2-1}\langle \Dec_K(\tau)\rangle\otimes \Lambda_2 - \frac{1}{|C|^2-1}\id^{\ch{M}} \otimes \Lambda_2\\
    &= \id^{\ch{M}} \otimes \Lambda_1 + \frac{1}{|C|^2-1}(|C|^2\lr{\Dec_K(\tau^{\ch{C}})} - \id)^{\ch{M}} \otimes \Lambda_2.
\end{align*}

From this it follows that
\[ \dn{\Tilde{\Lambda}_A - \left(\id^{\ch{M}} \otimes \Lambda_3^{\ch[B]{\hat{B}}} + \frac{1}{\alpha - 1}\left(\alpha\lr{\Dec_K(\tau^{\ch{C}})} - \id\right)^{\ch{M}} \otimes \Lambda_4^{\ch[B]{\hat{B}}}\right)} \leq \eps.\]
\end{proof}

\begin{theorem}[Theorem \ref{thm:pnm_charac}]
\label{thm:prf_pnm_charac}
Let $\Pi = (\KeyGen, \Enc, \Dec)$ be an arbitrary $\eps$-\PNM\ \SKQES\ for some $\eps$, then for any attack $\Lambda_A^{\ch[CB]{C\hat{B}}}$, its effective map $\Tilde{\Lambda}_A^{\ch[MB]{M\hat{B}}}$ is such that
\[ \dn{\Tilde{\Lambda}_A - \left(\id^{\ch{M}} \otimes \Lambda_1^{\ch[B]{\hat{B}}} + \frac{1}{|M|^2 - 1}\left(|M|^2\lr{\Dec_K(\tau^{\ch{C}})} - \id\right)^{\ch{M}} \otimes \Lambda_2^{\ch[B]{\hat{B}}}\right)} \leq 3\eps,\]
where
\begin{align*}
    \Lambda_1 &= \Trr[MM']{\phi^{+\ch{MM'}}\Tilde{\Lambda}_A(\phi^{+\ch{MM'}} \otimes (\cdot))}&and\\
    \Lambda_2 &= \Trr[MM']{(\I^{\ch{MM'}} - \phi^{+\ch{MM'}})\Tilde{\Lambda}_A(\phi^{+\ch{MM'}} \otimes (\cdot))}.
\end{align*}
\end{theorem}
\begin{proof}
Let $\Pi = (\KeyGen, \Enc, \Dec)$ be an arbitrary $\eps$-\PNM\ \SKQES\ for some $\eps$ and let $\Lambda_A^{\ch[CB]{C\hat{B}}}$ be an arbitrary attack with effective map $\Tilde{\Lambda}_A^{\ch[MB]{M\hat{B}}}$. Furthermore, let $\Lambda_1^{\ch[B]{\hat{B}}}$ and $\Lambda_2^{\ch[B]{\hat{B}}}$ be such that
\[ \dn{\Tilde{\Lambda}_A - \Tilde{\Lambda}_{ideal}} \leq \eps, \]
where $\Tilde{\Lambda}_{ideal}^{\ch[MB]{M\hat{B}}} = \id^{\ch{M}} \otimes \Lambda_1 + \frac{1}{|M|^2-1}(|M|^2\lr{\Dec_K(\tau)} - \id)^{\ch{M}} \otimes \Lambda_2$. Lastly, let
\begin{align*}
   \Lambda_3 &= \Trr[MM']{\phi^{+\ch{MM'}}\Tilde{\Lambda}_A(\phi^{+\ch{MM'}} \otimes (\cdot))}&&\text{,}\\
   \Lambda_4 &= \Trr[MM']{(\I^{\ch{MM'}} - \phi^{+\ch{MM'}})\Tilde{\Lambda}_A(\phi^{+\ch{MM'}} \otimes (\cdot))}&&\text{, and}\\
   \Tilde{\Lambda}_{trace}^{\ch[MB]{M\hat{B}}} &= \id^{\ch{M}} \otimes \Lambda_3 + \frac{1}{|M|^2-1}(|M|^2\lr{\Dec_K(\tau)} - \id)^{\ch{M}} \otimes \Lambda_4&&\text{.}
\end{align*}

Observe that, by the triangle inequality, $\dn{\Tilde{\Lambda} - \Tilde{\Lambda}_{trace}} \leq\dn{\Tilde{\Lambda} - \Tilde{\Lambda}_{ideal}} + \dn{\Tilde{\Lambda}_{ideal} - \Tilde{\Lambda}_{trace}}$. Furthermore,
\begin{align*}
    &\dn{\Tilde{\Lambda}_{ideal} - \Tilde{\Lambda}_{trace}}\\
    &= \dn{ \id^{\ch{M}} \otimes (\Lambda_1 - \Lambda_3) + \frac{1}{|M|^2-1}(|M|^2\lr{\Dec_K(\tau)} - \id)^{\ch{M}} \otimes (\Lambda_2 - \Lambda_4)}\\
    &\leq \dn{ \id^{\ch{M}} \otimes (\Lambda_1 - \Lambda_3)} + \dn{\frac{1}{|M|^2-1}(|M|^2\lr{\Dec_K(\tau)} - \id)^{\ch{M}} \otimes (\Lambda_2 - \Lambda_4)}\\
    &= \dn{ \id }\dn{ (\Lambda_1 - \Lambda_3)} + \dn{\frac{1}{|M|^2-1}(|M|^2\lr{\Dec_K(\tau)} - \id)^{\ch{M}} }\dn{ (\Lambda_2 - \Lambda_4)}\\
    &\leq \dn{ (\Lambda_1 - \Lambda_3)} + \dn{ (\Lambda_2 - \Lambda_4)}
\end{align*}

Let $\Lambda_5 = \Trr[MM']{\phi^{+\ch{MM'}}\Tilde{\Lambda}_{ideal}(\phi^{+\ch{MM'}} \otimes (\cdot))}$ and\\ $\Lambda_6 = \Trr[MM']{(\I^{\ch{MM'}} - \phi^{+\ch{MM'}})\Tilde{\Lambda}_{ideal}(\phi^{+\ch{MM'}} \otimes (\cdot))}$. Observe that the mapping
\[ \rho \mapsto \kb{0}{0} \otimes \Tr_{MM'}[\phi^{+\ch{MM'}}\rho] + \kb{1}{1} \otimes \Tr_{MM'}[(\I^{\ch{MM'}} - \phi^{+\ch{MM'}})\rho] \]
is \CPTP. Since $\dn{(\Tilde{\Lambda} - \Tilde{\Lambda}_{ideal})(\phi^{+\ch{MM'}} \otimes (\cdot))} \leq \dn{\Tilde{\Lambda} - \Tilde{\Lambda}_{ideal} } \leq \eps$ and the diamond norm is non-increasing under CPTP maps\footnote{See \cite{Watrous2018}, Proposition 3.48(1)
}, we have \\$\dn{\kb{0}{0} \otimes (\Lambda_3 - \Lambda_5) + \kb{1}{1} \otimes (\Lambda_4 - \Lambda_6) } \leq \eps$ and thus $\dn{\Lambda_3 - \Lambda_5} \leq \eps$ and $\dn{\Lambda_4 - \Lambda_6} \leq \eps$. Using this we observe that
\begin{align*}
    \dn{\Tilde{\Lambda}_{ideal} - \Tilde{\Lambda}_{trace} } &\leq\dn{ \Lambda_1 - \Lambda_3} + \dn{ \Lambda_2 - \Lambda_4}\\
    &\leq \dn{ \Lambda_1 - \Lambda_5} + \dn{ \Lambda_5 - \Lambda_3} +\dn{ \Lambda_2 - \Lambda_6} + \dn{ \Lambda_6 - \Lambda_4}\\
    &\leq 2\eps + \dn{ \Lambda_1 - \Lambda_5} + \dn{ \Lambda_2 - \Lambda_6}.
\end{align*}

Furthermore we have
\begin{align*}
    \Lambda_5 &= \Tr_{MM'}[\phi^{+}\Tilde{\Lambda}_{ideal}(\phi^{+} \otimes (\cdot))]\\
    &=\Trr[MM']{\phi^{+}\left(\phi^{+} \otimes \Lambda_1 + \frac{1}{|M|^2-1}(|M|^2\lr{\Dec_K(\tau)} - \id)(\phi^{+}) \otimes \Lambda_2\right)}\\
    &=\Trr[MM']{\phi^{+}\left(\phi^{+} \otimes \Lambda_1 + \frac{1}{|M|^2-1}(|M|^2\Dec_K(\tau) \otimes\tau^{\ch{M'}} - \phi^{+}) \otimes \Lambda_2\right)}\\
    &= \Lambda_1 +\Trr{\frac{1}{|M|^2-1}(|M|^2\phi^{+}(\Dec_K(\tau) \otimes\tau^{\ch{M'}}) - \phi^{+})}\Lambda_2\\
    &= \Lambda_1,
\end{align*}
where $\phi^+ = \phi^{+\ch{MM'}}$ and the last equality holds because
\begin{align*}
    \Trr{\phi^{+\ch{MM'}}(\Dec_K(\tau) \otimes\tau^{\ch{M'}})} &= \frac{1}{|M|}\Trr{\sum\limits_{i,j=0}^{|M|}\kb{ii}{jj}(\Dec_K(\tau) \otimes\tau^{\ch{M'}})}\\
    &= \frac{1}{|M|}\Trr{\sum\limits_{i,j=0}^{|M|}\kb{i}{j}\Dec_K(\tau) \otimes \kb{i}{j}\tau^{\ch{M'}})}\\
    &= \frac{1}{|M|^2}\Trr{\sum\limits_{i=0}^{|M|}\kb{i}{i}\Dec_K(\tau)}\\
    &= \frac{1}{|M|^2}
\end{align*}

Similarly
\begin{align*}
    \Lambda_6 &=\Trr[MM']{(\mathbb{I}^{\ch{MM'}} - \phi^{+\ch{MM'}})\Tilde{\Lambda}^{ideal}(\phi^{+\ch{MM'}} \otimes (\cdot))}\\
    &= \Trr[MM']{\Tilde{\Lambda}_{ideal}(\phi^{+\ch{MM'}} \otimes (\cdot))} - \Lambda_5\\
    &=\Trr[MM']{\left(\phi^{+} \otimes \Lambda_1 + \frac{1}{|M|^2-1}(|M|^2\Dec_K(\tau) \otimes\tau^{\ch{M'}} - \phi^{+}) \otimes \Lambda_2\right)}\\
    &= \Lambda_1 + \Lambda_2 - \Lambda_5\\
    &= \Lambda_2.
\end{align*}

From this we conclude
\begin{align*}
     \dn{\Tilde{\Lambda} - \Tilde{\Lambda}_{trace} } &\leq \dn{\Tilde{\Lambda} - \Tilde{\Lambda}_{ideal} } + \dn{\Tilde{\Lambda}_{ideal} - \Tilde{\Lambda}_{trace} }\\
     &\leq \eps + \dn{\Tilde{\Lambda}_{ideal} - \Tilde{\Lambda}_{trace} }\\
     &\leq 3\eps + \dn{ \Lambda_1 - \Lambda_5} + \dn{ \Lambda_2 - \Lambda_6}\\
     &= 3\eps.
\end{align*}
\end{proof}

\begin{theorem}[Theorem \ref{thm:pnm_dns}]
\label{thm:prf_pnm_dns}
For any $0 \leq \eps \leq 2$ and any $\eps$-\PNM\ \SKQES\ $\Pi = (\KeyGen, \Enc, \Dec)$, there exists some $x$ such that the scheme \\$\Pi' = (\KeyGen, \Enc', \Dec')$ is $\left(\frac{3}{|R|} + \eps\right)$-\DNS-authenticating, where
\begin{align*}
    \Enc'_k &= \Enc_k((\cdot)^{\ch{M'}} \otimes \kb{x}{x}^{\ch{R}})\\
    \Dec'_k &= \bra{x}^{\ch{R}}\Dec_k(\cdot)\ket{x}^{\ch{R}} + \Trr{(\I^{\ch{R}} - \kb{x}{x}^{\ch{R}})\Dec_k(\cdot)}\kb{\bot}{\bot}
\end{align*}
\end{theorem}
\begin{proof}
By Lemma~\ref{lem:dns_tag}, there exists an $x \in \{0,1\}^{\log|R|}$ such that\\ $\Trr{\bra{x}\Dec_K(\tau^{\ch{C}})\ket{x}} \leq \frac{1}{|R|}$. Fix this $x$ and define $\Pi'$ as above. Define $\Enc_{ap}(X) = X \otimes \kb{x}{x}$ and $\Dec_{ch}(Y) = \bra{x}Y\ket{x} + \Trr{(\I - \kb{x}{x})Y}\kb{\bot}{\bot}$ and observe that $\Enc' = \Enc \circ \Enc_{ap}$ and $\Dec' = \Dec_{ch} \circ \Dec$. Let $\Lambda_A$ be an arbitrary attack map on $\Pi'$, then its effective map is
\[ \Tilde{\Lambda}'_A = \E\limits_{k \leftarrow \KeyGen(1^n)}[\Dec'_k \circ \Lambda_A \circ \Enc'_k]. \]
Since $\Enc_{ap}$ and $\Dec_{ch}$ do not change with $k$ and are linear, we have
\[ \Tilde{\Lambda}'_A = \Dec_{ch} \circ \Tilde{\Lambda}_A \circ \Enc_{ap}, \]
where $\Tilde{\Lambda}_A = \E\limits_{k \leftarrow \KeyGen(1^n)}[(\Dec_k \circ \Lambda_A \circ \Enc_k)]$.
Since $\Pi$ is $\eps$-\PNM, there exist $\Lambda_1, \Lambda_2$ such that
\[ \dn{\Tilde{\Lambda}_A - \id \otimes \Lambda_1 + \frac{1}{|M|^2-1}(|M|^2\lr{\Dec_K(\tau^{\ch{C}})} - \id)^{\ch{M}}\otimes \Lambda_2}\leq \eps. \]

Since $\Enc_{ap}$ and $\Dec_{ch}$ are both \CPTP, by submultiplicativity we have that
\[ \dn{\Dec_{ch} \circ \left(\Tilde{\Lambda}_A - \id \otimes \Lambda_1 + \frac{1}{|M|^2-1}(|M|^2\lr{\Dec_K(\tau)} - \id)\otimes \Lambda_2\right) \circ \Enc_{ap}} \leq \eps, \]
which is equivalent to
\[ \dn{\Tilde{\Lambda}'_A - \Dec_{ch} \circ \left(\id \otimes \Lambda_1 + \frac{1}{|M|^2-1}(|M|^2\lr{\Dec_K(\tau)} - \id)\otimes \Lambda_2\right) \circ \Enc_{ap}} \leq \eps.\]
Observe that 
\[ \Dec_{ch} \circ \id \circ \Enc_{ap} = \bra{x}((\cdot) \otimes \kb{x}{x})\ket{x} + \Tr[(\I - \kb{x}{x})((\cdot) \otimes \kb{x}{x})]\kb{\bot}{\bot} = \id.\]
Define $\Lambda_{acc} = \Lambda_1$, $\Lambda_{rej} = \Lambda_2$ and \[\Tilde{\Lambda}_{ideal} = \Dec_{ch} \circ \left(\id \otimes \Lambda_1 + \frac{1}{|M|^2-1}(|M|^2\lr{\Dec_K(\tau^{\ch{C}})} - \id)^{\ch{M}}\otimes \Lambda_2\right) \circ \Enc_{ap},\] then we have
\begin{align*}
    &\dn{\Tilde{\Lambda}_{ideal} - \id \otimes \Lambda_{acc} - \lr{\kb{\bot}{\bot}} \otimes \Lambda_{rej}}\\
    = &\dn{\left(\frac{1}{|M|^2-1}(|M|^2(\Dec_{ch} \circ \lr{\Dec_K(\tau)} \circ \Enc_{ap}) - \id) - \lr{\kb{\bot}{\bot}}\right) \otimes \Lambda_2}\\
    \leq &\dn{\frac{1}{|M|^2-1}(|M|^2(\Dec_{ch} \circ \Trr{(\cdot)^{\ch{M'}}}\Dec_K(\tau)) - \id) - \lr{\kb{\bot}{\bot}}}
\end{align*}
Here the inequality uses the fact that $\Enc_{ap}$ is trace preserving and $\lr{\Dec_K(\tau)}$ is a constant channel, which only uses the trace of the input. Since every term ended in $\otimes \Lambda_2$, we removed this term and multiplied with $\dn{\Lambda_2}$, which is less than $1$ since $\Lambda_2$ is \CPTNI. We continue by expanding $\Dec_{ch}$, where we use that $\bra{x}\Trr{(\cdot)^{\ch{M'}}}\Dec_K(\tau^{\ch{C}})\ket{x} = \lr{\bra{x}\Dec_K(\tau^{\ch{C}})\ket{x}}$ and we abbreviate $\psi = \bra{x}\Dec_K(\tau^{\ch{C}})\ket{x}$ and $[\bot] = \lr{\kb{\bot}{\bot}}$.
\begin{align*}
    &\dn{\frac{1}{|M|^2-1}(|M|^2(\Dec_{ch} \circ \Trr{(\cdot)^{\ch{M'}}}\Dec_K(\tau^{\ch{C}})) - \id) - [\bot]}\\
    = &\dn{\frac{1}{|M|^2-1}(|M|^2\left(\lr{\psi} + \Tr[(\I - \kb{x}{x})\Dec_K(\tau^{\ch{C}})][\bot]\right) - \id)- [\bot]}.
\end{align*}
We can rewrite this expression by first rewriting $\Tr[(\I - \kb{x}{x})\Dec_K(\tau^{\ch{C}})]$ as $1 - \Tr[\psi]$, then collecting all multipliers of $\lr{\kb{\bot}{\bot}}$, and lastly distributing the $|M|^2$ term and simplifying the resulting term.
\begin{align*}
    &\dn{\frac{1}{|M|^2-1}(|M|^2\left(\lr{\psi} + (1 - \Tr[\psi])\lr{\kb{\bot}{\bot}}\right) - \id)- \lr{\kb{\bot}{\bot}}}\\
    = &\dn{\frac{1}{|M|^2-1}(|M|^2\left(\lr{\psi} + \left((1 - \Tr[\psi]) - \frac{|M|^2-1}{|M|^2}\right)\lr{\kb{\bot}{\bot}}\right) - \id)}\\
    = &\dn{\frac{1}{|M|^2-1}(|M|^2\lr{\psi} + \left(|M|^2(1 - \Tr[\psi]) - (|M|^2-1)\right)\lr{\kb{\bot}{\bot}} - \id)}\\
    = &\dn{\frac{1}{|M|^2-1}(|M|^2\lr{\psi} + \left(1 - |M|^2\Tr[\psi])\right)\lr{\kb{\bot}{\bot}} - \id)}\\
    \leq &\frac{1}{|M|^2-1}\left(|M|^2\dn{\lr{\psi}} + \dn{(1 - |M|^2\Tr[\psi])\lr{\kb{\bot}{\bot}}} + \dn{\id}\right)\\
    \leq &\frac{1}{|M|^2-1}\left(\frac{|M|^2}{|R|} + \left(\frac{|M|^2}{|R|} - 1\right) + 1\right)\leq \frac{3}{|R|}.
\end{align*}

Here the first inequality is an application of the triangle inequality. The second inequality uses the fact that $\dn{\id} = \dn{\lr{\kb{\bot}{\bot}}} = 1$ and that $|R|\lr{\psi}$ is \CPTNI\ because $\Trr{\bra{x}\Dec_K(\tau^{\ch{C}})\ket{x}} \leq \frac{1}{|R|}$ and thus $\dn{\lr{\psi}} \leq \frac{1}{|R|}$.

Since $\dn{\Tilde{\Lambda}'_A - \Tilde{\Lambda}_{ideal}} \leq \eps$ and $\dn{\Tilde{\Lambda}_{ideal} - \id \otimes \Lambda_{acc} - \lr{\kb{\bot}{\bot}} \otimes \Lambda_{rej}} \leq \frac{3}{|R|}$, we have by the triangle inequality that
\[ \dn{\Tilde{\Lambda}'_A - \id \otimes \Lambda_{acc} - \lr{\kb{\bot}{\bot}} \otimes \Lambda_{rej}} \leq \eps + \frac{3}{|R|},\]
which means that $\Pi'$ is $\left(\frac{3}{|R|} + \eps\right)$-\DNS\ authenticating.
\end{proof}

To prove \QCNM security of the classical-quantum hybrid scheme, we need the following lemma.

\begin{lemma}\label{lem:CiNM-multidecrypt}
	Let $\Pi=(\KeyGen, \Enc,\Dec)$ be a  \SKQES, let $\ell \in \N$, let $\mathbf{C} = C_1\dots C_\ell \cong C^\ell$ and $\mathbf{M} = M_1\dots M_\ell \cong M^\ell$ be vectors of registers, let $\Lambda^{\ch[C]{\mathbf{C}}}$ be a \CPTP map, and set
	\[
		\tilde\Lambda^{\ch[M]{\mathbf{M}}}=\E_{k\leftarrow\KeyGen(1^n)}\left[\left(\Dec_k\right)^{\otimes \ell}\circ\Lambda\circ\Enc_k\right].
	\]
	If \SKQES is \CiNM secure, then for some $p_0$ and $\{\sigma_i\}_i$ we have that
	\[
		\tilde\Lambda^{\ch[M]{\mathbf{M}}}=\sum_{i=1}^\ell p_i\id^{\ch[M]{M_i}}\otimes\sigma_i^{\ch{\mathbf{M}_{-i}}} + p_0 \lr{\sigma_0^{\ch{\mathbf{M}}}},
	\]
	where $q_i$ is the probability that is equal to $\Lambda_1$ from Definition \ref{def:cinm} applied to the attack map $\Tr_{\mathbf{C}_{-i}}\circ\Lambda$ and $\mathbf{M}_{-i} = M_1 \dots M_{i-1}M_{i+1} \dots M_\ell$ and $p_i = q_i - \frac{1}{|C|^2-1}(1-q_i)$.
\end{lemma}
\begin{proof}
	For fixed $i$, consider the attack $\Lambda' = \Tr_{\mathbf{C}_{-i}}\circ\Lambda$ on $\Pi$ and observe that its effective map satisfies $\tilde\Lambda' = \E_k\left[\Dec_k \circ \Tr_{\mathbf{C}_{-i}}\circ\Lambda\circ \Enc_k \right] = \Tr_{\mathbf{M}_{-i}}\circ\tilde\Lambda$. Because $\Pi$ is \CiNM, we have $\tilde\Lambda' = p_i \id + (1-p_i) \lr{\Dec_K(\tau)}$ and thus 
\begin{equation}\label{eq:oneoutput}
	\Tr_{\mathbf{M}_{-i}}\circ\tilde\Lambda(\phi^{+\ch{MM'}}) = p_i\phi^{+\ch{M_iM'}} + (1-p_i)\Dec_K(\tau) \otimes \tau^{\ch{M'}}.
\end{equation}
 Consider the state $\tilde\Lambda(\phi^{+\ch{MM'}})$. Because $\phi^{+\ch{M_iM'}}$ is a pure state, we know that $\tilde\Lambda(\phi^{+\ch{MM'}})$ is a convex combination of terms of the form $p_i\phi^{+\ch{M_iM'}}\otimes \sigma_{\mathbf M_{-i}}$ and a term $p_0 \sigma_0^{\ch{\mathbf{M}}} \otimes \tau^{\ch{M'}}$, i.e. 
 \begin{equation}
 	\tilde\Lambda(\phi^{+\ch{MM'}}) = \sum_{i=1}^\ell p_i\phi^{+\ch{M_iM'}}\otimes\sigma_i^{\ch{\mathbf{M}_{-i}}} + p_0 \sigma_0^{\ch{\mathbf{M}}} \otimes \tau^{\ch{M'}}.
 \end{equation}
  By the Choi-Jamiolkowski isomorphism\cite{jamiolkowski1972linear,choi1975completely} this means that 
	\[
		\tilde\Lambda^{\ch[M]{\mathbf{M}}}=\sum_{i=1}^\ell p_i\id^{\ch[M]{M_i}}\otimes\sigma_i^{\ch{\mathbf{M}_{-i}}} + p_0 \lr{\sigma_0^{\ch{\mathbf{M}}}}.
	\]
	Using Equation \eqref{eq:oneoutput}, we get in addition that all single-system marginals of $\sigma_i$, $i=0,...,\ell$ are equal to $\Dec_K(\tau)$. 
\end{proof}

Note that a similar statement can be proven for attack maps $\Lambda^{\ch[MB]{\mathbf M\tilde B}}$ with side information, but we only need the above statement in Theorem \ref{thm:security}.

\end{document}